\documentclass[12pt]{article}
\usepackage{graphpap, graphics, float, tabularx, array, setspace, natbib, multirow}
\usepackage{amsmath, amsthm, amsfonts, eucal}
\usepackage[multiple]{footmisc}

\long\def\symbolfootnote[#1]#2{\begingroup
\def\thefootnote{\fnsymbol{footnote}}\footnote[#1]{#2}\endgroup}
\numberwithin{equation}{section}
\newenvironment{proofLemma1}[1][Proof of Lemma \ref{dispersionLem}.]{\begin{trivlist}
\item[\hskip \labelsep {\bfseries #1}]}{\hfill\qed\end{trivlist}}
\newenvironment{proofLemma2}[1][Proof of Lemma \ref{alphaLem}.]{\begin{trivlist}
\item[\hskip \labelsep {\bfseries #1}]}{\hfill\qed\end{trivlist}}
\newenvironment{proofThm1}[1][Proof of Theorem \ref{relValueThm}.]{\begin{trivlist}
\item[\hskip \labelsep {\bfseries #1}]}{\hfill\qed\end{trivlist}}
\theoremstyle{plain}
\renewcommand{\baselinestretch}{1.3}

\theoremstyle{plain}
\newtheorem{thm}{Theorem}[section]
\newtheorem{lem}[thm]{Lemma}
\newtheorem{cor}[thm]{Corollary}
\theoremstyle{definition}
\newtheorem{defn}[thm]{Definition}

\def\ito{It{\^o}}
\def\R{{\mathbb R}}
\def\F{\CMcal{F}}
\def\Var{\operatorname{Var}}
\def\intT{\int_0^T}
\def\O{\Omega}
\def\D{\Delta}
\def\a{\alpha}
\def\b{\beta}

\def\l{\lambda}

\def\g{\gamma}

\begin{document}

\centerline{\bf \LARGE{Asset Price Distributions and Efficient Markets} }

\vskip 10pt




\vskip 5pt

\vskip 25pt
\centerline{\large{\hskip 10pt Ricardo T. Fernholz\symbolfootnote[1]{Robert Day School of Economics and Finance, Claremont McKenna College, 500 E. Ninth St., Claremont, CA 91711, rfernholz@cmc.edu} \hskip 95pt Caleb Stroup\symbolfootnote[2]{Department of Economics, Davidson College, 405 N. Main St., Davidson, NC 28035, castroup@davidson.edu}}}

\vskip 3pt

\centerline{\hskip 0pt Claremont McKenna College \hskip 85pt Davidson College}

\vskip 35pt

\centerline{\large{\today}}

\vskip 5pt


\vskip 50pt

\renewcommand{\baselinestretch}{1.1}
\begin{abstract}
We explore a decomposition in which returns on a large class of portfolios relative to the market depend on a smooth non-negative drift and changes in the asset price distribution. This decomposition is obtained using general continuous semimartingale price representations, and is thus consistent with virtually any asset pricing model. Fluctuations in portfolio relative returns depend on stochastic time-varying dispersion in asset prices. Thus, our framework uncovers an asset pricing factor whose existence emerges from an accounting identity universal across different economic and financial environments, a fact that has deep implications for market efficiency. In particular, in a closed, dividend-free market in which asset price dispersion is relatively constant, a large class of portfolios must necessarily outperform the market portfolio over time. We show that price dispersion in commodity futures markets has increased only slightly, and confirm the existence of substantial excess returns that co-vary with changes in price dispersion as predicted by our theory.
\end{abstract}
\renewcommand{\baselinestretch}{1.3}

\vskip 50pt



\vfill
\pagebreak

\section{Introduction} \label{sec:intro}

We explore the implications of two simple insights: First, a change in one asset's price relative to all other assets' prices must cause the distribution of relative asset prices to change. Second, an increase in the price of a relatively high-priced asset increases the dispersion of the asset price distribution. These two facts amount to accounting identities, and we show that they link the performance of a wide class of portfolios to the dynamic behavior of asset price dispersion.

To formalize and explore the relationship between asset price distributions, portfolio returns, and efficient markets, we represent asset prices as continuous semimartingales and show that returns on a large class of portfolios (relative to the market) can be decomposed into a drift and changes in the dispersion of asset prices:
\begin{equation} \label{intuitiveEq}
 \text{relative return} \;\; = \;\; \text{drift} \; - \; \text{change in asset price dispersion},
\end{equation}
where the drift is non-negative and roughly constant over time, asset price dispersion is volatile over time. Fluctuations in asset price dispersion thus drive fluctuations in portfolio returns relative to the market. The decomposition \eqref{intuitiveEq} is achieved using few assumptions about the dynamics of individual asset prices, which means our results are sufficiently general that they apply to almost any equilibrium asset pricing model, in a sense made precise in Section 2 below. Indeed, the decomposition \eqref{intuitiveEq} is little more than an accounting identity that is approximate in discrete time and exact in continuous time.

By tying the volatility of relative portfolio returns to changes in asset price dispersion, we characterize a class of asset pricing factors that are universal across different economic models and econometric specifications. Our results thus formally address concerns about the implausibly high number of factors and anomalies uncovered by the empirical asset pricing literature \citep{Novy-Marx:2014,Harvey/Liu/Zhu:2016,Bryzgalova:2016}. Indeed, the decomposition \eqref{intuitiveEq} provides a novel workaround to such criticisms, since its existence need not be rationalized by a particular equilibrium asset pricing model.

This generality is one of the strengths of our unconventional approach. The continuous semimartingales we use to represent asset prices allow for a practically unrestricted structure of time-varying dynamics and co-movements that are consistent with the endogenous price dynamics of any economic model. In this manner, our framework is applicable to both rational \citep{Sharpe:1964,Lucas:1978,Cochrane:2005} and behavioral \citep{Shiller:1981,DeBondt/Thaler:1989} theories of asset prices, as well as to the many econometric specifications of asset pricing factors identified in the empirical literature.

Our results also have implications for the efficiency of asset markets. We show that in a market in which dividends and the entry and exit of assets over time play small roles, non-negativity of the drift component of \eqref{intuitiveEq} implies that a wide class of portfolios must necessarily outperform the market except in the special case where asset price dispersion increases at a sufficiently fast rate over time. Thus, in order to rule out predictable excess returns, asset price dispersion must be increasing on average over time at a rate sufficiently fast to overwhelm the predictable positive drift. This result re-casts market efficiency in terms of a constraint on the dynamic behavior of asset price distributions. Through this lens, our approach provides a novel mechanism to uncover risk factors or inefficiencies that persist across a variety of different asset markets.

We test our theoretical predictions using commodity futures. This market provides a clear test of our theory, since commodity futures contracts do not pay dividends and rarely exit from the market. Although dividends and asset entry/exit can be incorporated into our framework and do not overturn the basic insight of the decomposition \eqref{intuitiveEq}, they do alter and complicate the form of our results. Furthermore, some of our results have been applied, albeit with a different interpretation, to equity markets \citep{Vervuurt/Karatzas:2015}, so the focus on commodity futures provides a completely novel set of empirical results that best aligns with our theoretical results.

In the decomposition \eqref{intuitiveEq}, asset price dispersion is any convex and symmetric function of relative asset prices, where each such function admits the decomposition for a specific portfolio. Our empirical analysis focuses on two special case measures of price dispersion, minus the geometric mean and minus the constant-elasticity-of-substitution (CES) function, and their associated equal- and CES-weighted portfolios. We decompose the relative returns of the equal- and CES-weighted commodity futures portfolios from 1974-2018 as in \eqref{intuitiveEq}, with the market portfolio defined as the price-weighted portfolio that holds one unit of each commodity futures contract. Empirically, we show that measures of commodity futures price dispersion increased only slightly over the forty year period we study. Consequently, and as predicted by the theory, the CES and equal-weighted portfolios exhibit positive long-run returns relative to the market driven by the accumulating positive drift. These portfolios consistently and substantially outperform the price-weighted market portfolio of commodity futures, with excess returns that have Sharpe ratios of 0.7-0.8 in most decades.


It is important to emphasize that the mathematical methods we use to derive our results are well-established and the subject of active research in statistics and mathematical finance. Our results and methods are most similar to \citet{Karatzas/Ruf:2017}, who provide a return decomposition similar to \eqref{intuitiveEq}. Their results are based on the original characterization of \citet{Fernholz:2002}, which led to subsequent contributions by \citet{Fernholz/Karatzas:2005}, \citet{Vervuurt/Karatzas:2015}, and \citet{Pal/Wong:2016}, among others. These contributions focus primarily on the solutions of stochastic differential equations and the mathematical conditions under which different types of arbitrage do or do not exist. Our results, in contrast, are expressed in terms of the distribution of asset prices and interpreted in an economic setting that focuses on questions of asset pricing risk factors and market efficiency. Ours is also the first paper to empirically examine these results in the commodity futures market, which, as discussed above, most closely aligns with the assumptions that underlie our theoretical results.

Our results raise the possibility of a unified interpretation of different asset pricing anomalies and risk factors in terms of the dynamics of asset price distributions. For example, the value anomaly for commodities uncovered by \citet{Asness/Moskowitz/Pedersen:2013} is similar in construction to the equal- and CES-weighted portfolios we study. Our theoretical results link these excess returns to fluctuations in commodity futures price dispersion and the approximate stability of this dispersion over long time periods. Therefore, any attempt to explain these excess returns or the related value anomaly for commodity futures of \citet{Asness/Moskowitz/Pedersen:2013} must also explain the dynamics of commodity price dispersion.


The same conclusion applies to the surprising finding of \citet{DeMiguel/Garlappi/Uppal:2009} that a naive strategy of weighting each asset equally --- $1/N$ diversification --- outperforms a variety of portfolio diversification strategies including a value-weighted market portfolio based on CAPM. The relative return decomposition \eqref{intuitiveEq} implies that such outperformance is likely a consequence of the approximate stability of asset price dispersion over time, given the positive drift. As with the value anomaly for commodities, then, any attempt to explain the relative performance of naive $1/N$ diversification must also explain the dynamics of asset price dispersion.

Our results also raise questions regarding the implications of equilibrium asset pricing models for price dispersion. Since our results show that asset price dispersion operates as a universal risk factor, different models' predictions for this dispersion become a major question of interest. In particular, our results imply that price dispersion should be linked to an endogenous stochastic discount factor that in equilibrium is linked to the marginal utility of economic agents. It is not obvious what economic and financial forces might underlie such a link, however. Nonetheless, our paper shows that unless asset price dispersion is consistently and rapidly rising over time, such links must necessarily exist.

\vskip 50pt



\section{Theory} \label{sec:theory}

In this section we ask what, if anything, can be learned about portfolio returns from information about the evolution of individual asset prices relative to each other. We do this by characterizing the close relationship between the distribution of relative asset prices and the returns for a large class of portfolios relative to the market. Importantly, our characterization is sufficiently broad as to nest virtually all equilibrium asset pricing theories, meaning that our results require no commitment to specific models of trading behavior, agent beliefs, or market microstructure.


\subsection{Setup and Discussion} \label{sec:setup}
Consider a market that consists of $N > 1$ assets. Time is continuous and denoted by $t$ and uncertainty in this market is represented by a probability space $(\O, \F, P)$ that is endowed with a right-continuous filtration $\{\F_t ; t \geq 0\}$. Each asset price $p_i$, $i = 1, \ldots, N$, is characterized by a positive continuous semimartingale that is adapted to $\{\F_t ; t \geq 0\}$, so that
\begin{equation} \label{contSemimart}
 p_i(t) = p_i(0) + g_i(t) + v_i(t),
\end{equation}
where $g_i$ is a continuous process of finite variation, $v_i$ is a continuous, square-integrable local martingale, and $p_i(0)$ is the initial price. The semimartingale representation \eqref{contSemimart} decomposes asset price dynamics into a time-varying cumulative growth component, $g_i(t)$, whose total variation is finite over every interval $[0, T]$, and a randomly fluctuating local martingale component, $v_i(t)$. By representing asset prices as continuous semimartingales, we are able to impose almost no structure on the underlying economic environment.

For any continuous semimartingales $x, y$, let $\langle x, y \rangle$ denote the cross variation of these processes and $\langle x \rangle = \langle x, x \rangle$ denote the quadratic variation of $x$. Since the continuous semimartingale decomposition \eqref{contSemimart} is unique and the finite variation processes $g_i$ and $g_j$ all have zero cross-variation \citep{Karatzas/Shreve:1991}, it follows that
\begin{equation} \label{crossVariation}
 \langle p_i, p_j \rangle (t) = \langle v_i, v_j \rangle (t),
\end{equation}
for all $i, j = 1, \ldots, N$ and all $t$. The cross-variation processes $\langle p_i, p_j \rangle$ measure the cumulative covariance between asset prices $p_i$ and $p_j$, and thus the differentials of these processes, $d \langle p_i, p_j \rangle (t)$, measure the instantaneous covariance between $p_i$ and $p_j$ at time $t$. Similarly, the quadratic variation processes $\langle p_i \rangle$ measure the cumulative variance of the asset price $p_i$, and thus the differential of that process, $d \langle p_i \rangle (t)$, measures the instantaneous variance of $p_i$.



Our approach in this paper is unconventional in that we do not impose a specific model of asset pricing. Instead, we derive results in a very general setting, with the understanding that the minimal assumptions behind these results mean that they will be consistent with almost any underlying economic model. Indeed, essentially any asset price dynamics generated endogenously by a model can be represented as general continuous semimartingales of the form \eqref{contSemimart}. This generality is crucial, since we wish to provide results that apply to all economic and financial environments.

Before proceeding, we pause and ask what assumptions the framework \eqref{contSemimart} relies on and how these assumptions relate to other asset pricing theories. A first important assumption is that assets do not pay dividends, so that returns are driven entirely by capital gains via price changes. In this sense, we can think of the $N$ assets in the market as rolled-over futures contracts that guarantee delivery of some underlying real asset on a future date. We emphasize that this assumption is for simplicity only. Our results can easily be extended to include general continuous semimartingale dividend processes similar to \eqref{contSemimart}. Including such dividend processes complicates the theory but does not change the basic insight of our results, a point that we discuss further below.

Second, we consider a closed market in which there is no asset entry or exit over time. In other words, we assume that the $N$ assets in the market are unchanged over time. As with dividends, our basic framework can be extended to include asset entry and exit using local time processes that measure the intensity of crossovers in rank (see, for example, \citet{Fernholz:2017a}). In such an extension, only the top $N$ assets in the market at a given moment in time are considered, and there is a local time process that measures the impact of entry and exit into and out of that top $N$, just like in the framework of \citet{Fernholz/Fernholz:2018}. For simplicity, we do not include asset entry and exit and the requisite local time processes in our theoretical analysis. We do, however, discuss how such entry and exit might impact our results. Furthermore, for our empirical analysis in Section \ref{sec:empirics} we consider commodity futures contracts in which asset exit --- the more significant omission --- does not occur over our sample period, thus aligning our empirical analysis as closely as possible with the theoretical assumptions of no dividends and no entry or exit.

In addition, we emphasize that our assumption that the continuous semimartingale price processes \eqref{contSemimart} are positive is only for simplicity and can be relaxed. Indeed, \citet{Karatzas/Ruf:2017} show how many of our theoretical results can be extended to a market in which zero prices are possible. Since zero prices are essentially equivalent to asset exit, this extension of our results provides another example of how the exit of assets from the market can be incorporated into our framework without overturning the basic insight of our results.

The last assumption behind \eqref{contSemimart} is that the prices $p_i$ are continuous functions of time $t$ that are adapted to the filtration $\{\F_t ; t \geq 0\}$. The assumption of continuity is essential for mathematical tractability, since we rely on stochastic differential equations whose solutions are readily obtainable in continuous time to derive our theoretical results. Given the generality of our setup, it is difficult to see how introducing instantaneous jumps into the asset price dynamics \eqref{contSemimart} would meaningfully alter our conclusions. Nonetheless, it is important and reassuring that in Section \ref{sec:empirics} we confirm the validity of our continuous-time results using monthly, discrete-time asset price data. This is not surprising, however, since an instantaneous price jump is indistinguishable from a rapid but continuous price change --- this is allowed according to \eqref{contSemimart} --- using discrete-time data. Finally, the assumption that asset prices $p_i$ are adapted means only that they cannot depend on the future. This reflects the reality that agents are not clairvoyant, and cannot relay information about the future realization of stochastic processes to the present.

The decomposition \eqref{contSemimart} separates asset price dynamics into two distinct parts. The first, the finite variation process $g_i$, has an instantaneous variance of zero (zero quadratic variation) and measures the cumulative growth in price over time. Despite its finite variation, the cumulative growth process $g_i$ can constantly change depending on economic and financial conditions as well as other factors, including the prices of the different assets. In Section \ref{sec:generalPort}, we show that our main relative return decomposition result consists of a finite variation process as well. In the subsequent empirical analysis Section \ref{sec:empirics}, we construct this finite variation process using discrete-time asset price data and show a clear contrast between its time-series behavior and the behavior of processes with positive quadratic variation (instantaneous variance greater than zero). 

The second part of the decomposition \eqref{contSemimart} consists of the square-integrable local martingale $v_i$. In general, this process has a positive instantaneous variance (positive quadratic variation), and thus its fluctuations are much larger and faster than for the finite variation cumulative growth process $g_i$. Note that a local martingale is more general than a martingale \citep{Karatzas/Shreve:1991}, and thus includes an extremely broad class of continuous stochastic processes. Intuitively, the process $v_i$ can be thought of as a random walk with a variance that can constantly change depending on economic and financial conditions as well as other factors. Furthermore, we allow for a rich structure of potentially time-varying covariances among the local martingale components $v_i$ of different asset prices, which are measured by the cross-variation processes \eqref{crossVariation}.

The commodity futures market we apply our theoretical results to in Section \ref{sec:empirics} offers one of the cleanest applications of our theory, since commodities rarely exit the market and their futures contracts do not pay dividends. A number of studies have decomposed commodity futures prices into risk premia and forecasts of future spot prices \citep{Fama/French:1987,Chinn/Coibion:2014}. Commodity spot prices, which are a major determinant of futures prices, have in turn been linked to storage costs and fluctuations in supply and demand \citep{Brennan:1958,Alquist/Coibion:2014}. In the context of this literature, there are many potential mappings from the fundamental economic and financial forces that determine spot and futures commodity prices to the general continuous semimartingale representation of asset prices \eqref{contSemimart}. Indeed, higher storage costs, increases in demand, rising risk premia, and many other factors can be represented as increases in the cumulative growth process $g_i$. Similarly, all of the unpredictable random shocks that impact commodity markets can be represented as changes in the local martingale $v_i$.

The crucial point, however, is that all of these models and the different economic and financial factors that they emphasize are consistent with the reduced form representation of asset prices \eqref{contSemimart}. After all, any model that proposes an explanation for the growth and volatility of commodity prices can be translated into our setup. The advantage of \eqref{contSemimart} is that we need not commit to any particular model of asset pricing, thus allowing us to derive results that are consistent across all the different models.



\subsection{Portfolio Strategies} \label{sec:port}
A \emph{portfolio strategy} $s(t) = (s_1(t), \ldots, s_N(t))$ specifies the number of shares of each asset $i = 1, \ldots, N$ that are to be held at time $t$. The shares $s_1, \ldots, s_N$ that make up a portfolio strategy must be measurable, adapted, and non-negative.\footnote{The assumption that portfolios hold only non-negative shares of each asset, and hence do not hold short positions, is only for simplicity. Our theory and results can be extended to long-short portfolios as well.} The \emph{value} of a portfolio strategy $s$ is denoted by $V_s > 0$, and satisfies
\begin{equation} \label{valueEq}
 V_s(t) = \sum_{i=1}^N s_i(t)p_i(t),
\end{equation}
for all $t$.

It is sometimes also useful to describe portfolio strategies $s$ in terms of \emph{weights}, denoted by $w^s(t) = (w^s_1(t), \ldots, w^s_N(t))$, which measure the fraction of portfolio $s$ invested in each asset. The shares of each asset held by a portfolio strategy, $s_i$, are easily linked to the weights of that portfolio strategy, $w^s_i$. In particular, a portfolio strategy $s(t) = (s_1(t), \ldots, s_N(t))$ has weights equal to
\begin{equation} \label{weightsEq}
 w^s_i(t) = \frac{p_i(t)s_i(t)}{V_s(t)},
\end{equation}
for all $i = 1, \ldots, N$ and all $t$, since \eqref{weightsEq} is equal to the dollar value invested in asset $i$ by portfolio $s$ divided by the dollar value of portfolio $s$. It is easy to confirm using \eqref{valueEq} and \eqref{weightsEq} that the weights $w^s_i$ sum up to one.

We require that all portfolios satisfy the self-financibility constraint, which ensures that gains or losses from the portfolio strategy $s$ account for all changes in the value of the investment over time. This implies that
\begin{equation} \label{selfFinanceEq}
 V_s(t) - V_s(0) = \int_0^t \sum_{i=1}^N s_i(t)\,dp_i(t),
\end{equation}
for all $t$. In addition, in order to permit comparisons on an even playing field, we set the initial holdings for all portfolios equal to each other. Without loss of generality, we set this initial value equal to the combined initial price of all assets in the economy, so that
\begin{equation} \label{valueNormalizationEq}
 V_s(0) = \sum_{i=1}^N p_i(0),
\end{equation}
for all portfolio strategies $s$.

One simple example of a portfolio strategy that will play a central role in much of our theoretical and empirical analysis is the \emph{market portfolio strategy}, which we denote by $m$. The market portfolio $m$ holds one share of each asset, so that $m(t) = (1, \ldots, 1)$ for all $t$. Following \eqref{valueEq}, we have that the value of the market portfolio strategy, $V_m$, is given by
\begin{equation} \label{marketValueEq}
  V_m(t) = \sum_{i=1}^N m_i(t)p_i(t) = \sum_{i=1}^N p_i(t),
\end{equation}
for all $t$. Note that the market portfolio satisfies the self-financibility constraint, since
\begin{equation}
 \int_0^t \sum_{i=1}^N m_i(t)\,dp_i(t) = \int_0^t \sum_{i=1}^N \,dp_i(t) = \sum_{i=1}^N p_i(t) - \sum_{i=1}^N p_i(0) = V_m(t) - V_m(0),
\end{equation}
for all $t$. It also satisfies the initial condition \eqref{valueNormalizationEq}, as shown by evaluating \eqref{marketValueEq} at $t = 0$.

Our definition of a portfolio strategy is very broad and includes many strategies that would be difficult or costly to implement in the real world. This broadness is intentional, as it helps to showcase the generality and power of our theoretical results. Indeed, one of our main contributions is to show that the returns for a large class of portfolios can be characterized parsimoniously under almost no assumptions about the underlying dynamics of asset prices. Once we have established this decomposition for general portfolio strategies, we will turn to specific examples to explain and highlight our results.


\subsection{The Distribution of Asset Prices} \label{sec:distribution}
The distribution of asset prices in our framework can be described in a simple way as a function of relative prices. Let $\theta = (\theta_1, \ldots, \theta_N)$, where each $\theta_i$, $i = 1, \ldots, N$, is given by
\begin{equation} \label{relPricesEq}
 \theta_i(t) = \frac{p_i(t)}{\sum_{i=1}^N p_i(t)}.
\end{equation}
Because the continuous semimartingales $p_i$ are all positive by assumption, it follows that $0 < \theta_i < 1$, for all $i = 1, \ldots, N$. By construction, we also have that $\theta_1 + \cdots + \theta_N = 1$. We denote the range of the relative price vector $\theta = (\theta_1, \ldots, \theta_N)$ by $\D$, so that
\begin{equation} \label{delta}
 \D = \left\{ (\theta_1, \ldots, \theta_N) \in (0, 1)^N : \sum_{i=1}^N \theta_i = 1 \right\}.
\end{equation}

Note that the market portfolio strategy $m$, which is defined as holding one share of each asset at all times, has weights equal to the relative price vector $\theta$. This is an immediate consequence of \eqref{weightsEq} and \eqref{marketValueEq}, which together imply that
\begin{equation} \label{marketWeights}
 w^m_i(t) = \frac{p_i(t)}{V_m(t)} = \frac{p_i(t)}{\sum_{i=1}^N p_i(t)} = \theta_i(t),
\end{equation}
for all $i = 1, \ldots, N$ and all $t$.

The portfolio strategies we characterize are constructed using measures of the dispersion of the asset price distribution. We demonstrate that the returns on these portfolios relative to the market portfolio depend crucially on changes in this asset price dispersion. The following definition makes dispersion of the asset price distribution a precise concept.

\begin{defn} \label{dispersionDef}
A twice continuously differentiable function $F : \D \to \R$ is a \emph{measure of price dispersion} if it is convex and invariant under permutations of the relative asset prices $\theta_1, \ldots, \theta_N$.
\end{defn}

We say that asset prices are more (less) dispersed as a measure of price dispersion $F$ increases (decreases). The following lemma explains why Definition \ref{dispersionDef}, which is the convex analogue of the diversity measure from \citet{Fernholz:2002}, accurately captures the concept of asset price dispersion.

\begin{lem} \label{dispersionLem}
Let $F$ be a measure of price dispersion and $\theta, \theta' \in \D$. Suppose that
\begin{equation}
 \max(\theta) = \max(\theta_1, \ldots, \theta_N) > \max(\theta'_1, \ldots, \theta'_N) = \max(\theta'),
\end{equation}
and that $\theta_i = \theta'_i$ for all $i$ in some subset of $\{1, \ldots, N\}$ that contains $N - 2$ elements. Then $F(\theta) \geq F(\theta')$. Furthermore, if $F$ is strictly convex, then $F(\theta) > F(\theta')$.
\end{lem}

To see how Lemma \ref{dispersionLem} explains the validity of Definition \ref{dispersionDef}, let us consider two relative price vectors $\theta, \theta' \in \D$. Suppose that the maximum relative price for $\theta$ is greater than for $\theta'$, while all other relative prices but one are equal to each other.\footnote{Note that if $\max(\theta) > \max(\theta')$, then it must be that $\theta_i \neq \theta'_i$ for at least two indexes $i = 1, \ldots, N$.} In this case, the relative price vector $\theta$ is more dispersed than $\theta'$, since these two are equal except that $\theta$ has a higher maximum price than $\theta'$. According to Lemma \ref{dispersionLem}, any measure of price dispersion $F$ will be weakly greater for $\theta$ than for $\theta'$ in this case, thus demonstrating that $F$ is weakly increasing in asset price dispersion. By a similar logic, the lemma also establishes that any strictly convex measure of price dispersion $F$ is strictly increasing in asset price dispersion.


We wish to consider two specific measures of price dispersion, both of which play a crucial role in forming portfolios for our empirical analysis. The first measure is based on the \emph{geometric mean function} $G : \D \to [0, \infty)$, defined by
\begin{equation} \label{geometricMeanEq}
 G(\theta(t)) = \left( \theta_1(t)\cdots\theta_N(t) \right)^{1/N}.
\end{equation}
Because the geometric mean function is concave, the function $-G < 0$ is a measure of price dispersion according to Definition \ref{dispersionDef}. The second measure of price dispersion is based on the \emph{constant elasticity of substitution (CES) function} $U : \D \to [0, \infty)$, defined by
\begin{equation} \label{cesEq}
  U(\theta(t)) = \left( \sum_{i=1}^N \theta^{\g}_i(t) \right)^{1/\g},
\end{equation}
where $\g$ is a nonzero constant, for all $i = 1, \ldots, N$.\footnote{For simplicity, we rule out the case where $\g = 0$ and $U$ becomes a Cobb-Douglas function.} As with the geometric mean function, the CES function is also concave, and hence the function $-U < 0$ is a measure of price dispersion according to Definition \ref{dispersionDef}.

The portfolio strategies that we construct using measures of price dispersion have relative returns that can be decomposed into changes in asset price dispersion and a non-negative drift process. This drift process is defined in terms of an associated measure of price dispersion. For any measure of price dispersion $F$ and any $i, j = 1, \ldots, N$, let $F_i$ denote the partial derivative of $F$ with respect to $\theta_i$, $\frac{\partial F}{\partial \theta_i}$,  let $F_{ij}$ denote the partial derivative of $F$ with respect to $\theta_i$ and $\theta_j$, $\frac{\partial^2 F}{\partial \theta_i \partial \theta_j}$, and let $H_F = (F_{ij})_{1 \leq i,j \leq N}$ denote the Hessian matrix of $F$.

\begin{defn} \label{driftDef}
For any measure of price dispersion $F$, the associated \emph{drift process} $\a_F$ is given by
\begin{equation} \label{alpha}
 \a_F(\theta(t)) =  \frac{1}{2}\sum_{i, j = 1}^N F_{ij}(\theta(t))\,d \langle \theta_i, \theta_j \rangle (t).
\end{equation}
\end{defn}

\begin{lem} \label{alphaLem}
For any measure of price dispersion $F$, the drift process $\a_F$ satisfies $\a_F \geq 0$. Furthermore, if $\operatorname{rank}(H_F) > 1$ and the covariance matrix $\left(d \langle p_i, p_j \rangle \right)_{1 \leq i, j \leq N}$ is positive definite for all $t$, then $\a_F > 0$.
\end{lem}

The non-negativity of the drift process $\a_F$ is significant. We show that, together with changes in price dispersion, this process accurately describes the returns of a large class of portfolio strategies relative to the market via a decomposition of the form \eqref{intuitiveEq}. Thus, if asset price dispersion is roughly unchanged over long time periods, then the long-run relative returns for many portfolios will be dominated by the drift process and hence will be non-negative. Furthermore, in this scenario the long-run relative return will be strictly positive if the measure of price dispersion $F$ is chosen appropriately ---  so that $\operatorname{rank}(H_F) > 1$ --- and the instantaneous variance of asset prices is positive and not perfectly correlated --- so that $\left(d \langle p_i, p_j \rangle \right)_{1 \leq i, j \leq N}$ is positive definite. In fact, in Section \ref{sec:empirics} we confirm the approximate stability of commodity futures price dispersion over long time periods and the predictable positive relative returns that this stability implies.

\subsection{General Results} \label{sec:generalPort}
In this section, we characterize the returns for a broad class of portfolio strategies relative to the market. We show that these relative returns can be decomposed into the non-negative drift process defined in the previous section and changes in price dispersion, just like in \eqref{intuitiveEq}.

One of the key ideas that underlies our results is that each measure of price dispersion $F$ has a corresponding portfolio strategy whose returns relative to the market are characterized by the value of the associated non-negative drift process $\a_F$ and changes in $F$. For this reason, measures of price dispersion are commonly said to ``generate'' the corresponding portfolio that admits such a decomposition \citep{Fernholz:2002,Karatzas/Ruf:2017}. One implication of this result is that there is a one-to-one link between measures of price dispersion and portfolio strategies whose relative returns depend on changes in that measure of price dispersion. The following theorem, which is similar to the more general results in Proposition 4.7 of \citet{Karatzas/Ruf:2017}, formalizes this idea.

\begin{thm} \label{relValueThm}
Let $F$ be a measure of price dispersion, and suppose that $F(\theta) < 0$ for all $\theta \in \D$. Then, the portfolio strategy $s(t) = (s_1(t), \ldots, s_N(t))$ with
\begin{equation} \label{strategyEq}
 s_i(t) = \frac{V_s(t)}{V_m(t)}\left( 1 + \frac{1}{F(\theta(t))}\left( F_i(\theta(t)) - \sum_{j=1}^N \theta_j(t)F_j(\theta(t)) \right) \right),
\end{equation}
for each $i = 1, \ldots, N$, has a value process $V_s$ that satisfies\footnote{Note that the stochastic integral $\int \a_F$ is evaluated with respect to the cross variation processes contained in $\a_F$, according to \eqref{alpha}.}
\begin{equation} \label{relValueEq}
 \log V_s(T) - \log V_m(T) = -\int_0^T\frac{\a_F(\theta(t))}{F(\theta(t))} + \log ( -F(\theta(T)) ),
\end{equation}
for all $T$.
\end{thm}

Theorem \ref{relValueThm} is powerful because it decomposes the returns for a broad class of portfolio strategies into the cumulative value of the non-negative drift process $\a_F$ and price dispersion as measured by $F$.\footnote{In Appendix \ref{supp}, we show that it is not necessary to characterize the decomposition in \eqref{strategyEq} in terms of logarithms. See Theorem \ref{relValueThmApp}.} Crucially, these portfolio strategies are easily implemented \emph{without any knowledge of the underlying fundamentals of the assets}. The portfolio $s$ of \eqref{strategyEq} specifies a number of shares of each asset to hold at time $t$ as a function of the prices of different assets relative to each other at time $t$, as measured by the relative price vector $\theta(t)$, and the relative value of the portfolio at time $t$, as measured by $V_s(t)/V_m(t)$. These quantities are easily observed over time, and do not require difficult calculations or costly information acquisition.

The decomposition \eqref{relValueEq} from Theorem \ref{relValueThm} characterizes the log value of the portfolio strategy $s$ relative to the log value of the market portfolio strategy $m$ at time $T$ in terms of the cumulative value of the associated drift process adjusted by price dispersion, $-\int_0^T\frac{\a_F(\theta(t))}{F(\theta(t))}$, and the log value of minus asset price dispersion, $\log ( -F(\theta(T)) )$. In order to go from this characterization of relative portfolio values to a characterization of relative portfolio returns, we take differentials of both sides of \eqref{relValueEq}. This yields
\begin{equation} \label{relReturnEq}
 d\log V_s(t) - d\log V_m(t) =  -\frac{\a_F(\theta(t))}{F(\theta(t))} +  d\log ( -F(\theta(t)) ),
\end{equation}
for all $t$. According to \eqref{relReturnEq}, then, the log return of the portfolio $s$ relative to the market can be decomposed into the non-negative value of the drift process adjusted by price dispersion, measured by $-\a_F/F \geq 0$, and changes in asset price dispersion, measured by $d\log ( -F )$.

The relative return characterization \eqref{relReturnEq} is of the same form as the intuitive version \eqref{intuitiveEq} presented in the Introduction. Therefore, Theorem \ref{relValueThm} implies that increases (decreases) in asset price dispersion lower (raise) the relative returns on a large class of portfolios. It also implies that if price dispersion is unchanged, then the relative returns on this large class of portfolios will be either non-negative or positive, since the drift process from \eqref{relValueEq} is either non-negative or positive according to Lemma \ref{alphaLem}. We confirm both of these predictions using commodity futures data in Section \ref{sec:empirics}.

Another implication of Theorem \ref{relValueThm} is that one part of the decomposition of the relative value of the portfolio strategy $s$ is a finite variation process. In particular, the cumulative value of the drift process adjusted by price dispersion, $-\int_0^T\frac{\a_F(\theta(t))}{F(\theta(t))}$, is a finite variation process by construction. To see why, note that the stochastic integral of a non-negative continuous process is continuous and non-decreasing, and any non-decreasing continuous process is a finite variation process \citep{Karatzas/Shreve:1991}.

Recall from Section \ref{sec:setup} that a finite variation process has finite total variation over every interval $[0, T]$. This means that the process has zero quadratic variation, or equivalently, zero instantaneous variance. In Section \ref{sec:empirics}, we decompose actual relative returns as described by Theorem \ref{relValueThm} using monthly commodity futures data and show a clear contrast between the time-series behavior of the zero-instantaneous-variance process $-\int_0^T\frac{\a_F(\theta(t))}{F(\theta(t))}$ and the positive-instantaneous-variance process $\log ( -F(\theta(T)) )$. In particular, we find that the sample variance of the finite variation process is orders of magnitude lower than that of the positive quadratic variation process, as predicted by the theorem.

The decompositions \eqref{relValueEq} and \eqref{relReturnEq} are little more than accounting identities, which are approximate in discrete time and exact in continuous time. There are essentially no restrictive assumptions about the underlying dynamics of asset prices and their co-movements that go into these results, making it difficult to imagine an equilibrium model of asset pricing that meaningfully clashes with Theorem \ref{relValueThm}. Despite this generality, two simplifying assumptions behind these results --- that assets do not pay dividends, and that the market is closed so that there is no asset entry or exit over time --- merit further discussion.

If we were to include dividends in our framework, we would get a relative value decomposition that is very similar to \eqref{relValueEq}. The only difference in this case would be an extra term added to \eqref{relValueEq} measuring the cumulative dividends from the portfolio strategy $s$ relative to the cumulative dividends from the market portfolio strategy. In the presence of dividends, then, relative capital gains could still be decomposed into the drift process and changes in price dispersion as in Theorem \ref{relValueThm}. The only complication would be an extra term that measures relative cumulative dividends as part of relative investment value.

The result of Theorem \ref{relValueThm} can also be extended to include asset entry and exit over time. As discussed in Section \ref{sec:setup}, the closed market assumption can be relaxed by introducing a local time process that measures the impact of asset entry and exit to and from the market, as detailed by \citet{Fernholz/Fernholz:2018}. If we were to relax this assumption and include asset entry and exit in our framework, we would get a relative value decomposition that is identical to \eqref{relReturnEq} plus one extra term that measures the differential impact of entry and exit on the returns of the portfolio strategy $s$ versus the market portfolio strategy. As with dividends, then, relative returns could still be decomposed into the drift process and changes in price dispersion as in Theorem \ref{relValueThm} in this case. The only complication would be an extra term that measures the relative impact of entry and exit on returns.

\subsection{Proof Sketch} \label{sec:proof}
We present the proof of Theorem \ref{relValueThm} in Appendix \ref{proofs}. In this subsection, we provide a sketch of this proof using second-order Taylor approximations of functions, with the understanding that these approximations are exact in continuous time by \ito's lemma \citep{Karatzas/Shreve:1991,Nielsen:1999}. Furthermore, for any function $f$, we use the notation $df(x)$ and $d^2f(x)$ to denote, respectively, $f(x) - f(x')$ and $(f(x) - f(x'))^2$, where $x - x' \in \R^N$ is approximately equal to $(0, \ldots, 0)$.

Let $F < 0$ be a measure of price dispersion, which has a second-order Taylor approximation given by
\begin{equation} \label{pfSketchEq1}
 dF(\theta(t)) \approx \sum_{i=1}^N F_i(\theta(t)) \, d\theta_i(t) + \frac{1}{2}\sum_{i,j =1}^N F_{ij}(\theta(t)) \, d\theta_i(t) \, d\theta_j(t).
\end{equation}
Consider a portfolio strategy $s$ that holds shares
\begin{equation} \label{pfSketchEq2}
 s_i(t) = \frac{V_s(t)}{V_m(t)}\left(c(t) + \frac{F_i(\theta(t))}{F(\theta(t))}\right),
\end{equation}
for each $i = 1, \ldots, N$, where $c(t)$ is potentially time-varying and sets $\sum_{i=1}^N p_i(t)s_i(t) = V_s(t)$ for all $t$, thus ensuring that $s$ is a valid portfolio strategy according to \eqref{valueEq}. Note that \eqref{pfSketchEq2} defines the portfolio strategy $s$ in the same way as \eqref{strategyEq} in Theorem \ref{relValueThm}. Following \eqref{selfFinanceEq}, for this proof sketch we assume that
\begin{equation} \label{pfSketchEq3}
 d\frac{V_s(t)}{V_m(t)} = \sum_{i=1}^Ns_i(t) \, d\theta_i(t),
\end{equation}
and we leave the derivation of this equation to Appendix \ref{proofs}. Substituting \eqref{pfSketchEq2} into \eqref{pfSketchEq3} yields
\begin{equation} \label{pfSketchEq4}
 d\frac{V_s(t)}{V_m(t)} = \frac{V_s(t)}{V_m(t)} \sum_{i=1}^N \left(c(t) + \frac{F_i(\theta(t))}{F(\theta(t))}\right)  d\theta_i(t) = \frac{V_s(t)}{V_m(t)} \sum_{i=1}^N\frac{F_i(\theta(t))}{F(\theta(t))} \, d\theta_i(t),
\end{equation}
for all $t$, where the last equality follows from the fact that $c(t)$ does not vary across different $i$ and $\sum_{i=1}^N \theta_i(t) = 1$ for all $t$, so that $d \sum_{i=1}^N\theta_i(t) = 0$. If we substitute \eqref{alpha} and \eqref{pfSketchEq4} into \eqref{pfSketchEq1}, then we have
\begin{equation} \label{pfSketchEq5}
 \frac{d(V_s(t)/V_m(t))}{V_s(t)/V_m(t)}  \approx -\frac{\a_F(\theta(t))}{F(\theta(t))} + \frac{dF(\theta(t))}{F(\theta(t))},
\end{equation}
for all $t$.

Let $\O$ be the process
\begin{equation} \label{omega}
 \O(\theta(t)) = -F(\theta(t)) \exp\int_0^t-\frac{\a_F(\theta(s))}{F(\theta(s))}.
\end{equation}
According to \ito's product rule \citep{Karatzas/Shreve:1991}, the second-order Taylor approximation of $\O$ is given by\footnote{For an informal derivation of \ito's product rule, note that the second-order Taylor approximation of $f(x_1, x_2) = x_1x_2$ is given by
\begin{equation*}
 df = d (x_1x_2) \approx x_1 \, dx_2 + x_2 \, dx_1 + dx_1 \, dx_2.
\end{equation*}}
\begin{equation} \label{pfSketchEq6}
\begin{aligned}
 d \O(\theta(t)) & \approx -dF(\theta(t)) \, \exp\int_0^t-\frac{\a_F(\theta(s))}{F(\theta(s))} + \a_F(\theta(t))\exp\int_0^t-\frac{\a_F(\theta(s))}{F(\theta(s))}  \\
 & \qquad \qquad \qquad \qquad \qquad + \; d(-F(\theta(t))) \, d\left( \exp\int_0^t-\frac{\a_F(\theta(s))}{F(\theta(s))}  \right),
\end{aligned}
\end{equation}
for all $t$. As discussed above, the stochastic integral $\int_0^t\frac{\a_F(\theta(s))}{F(\theta(s))} $ is a finite variation process, and therefore the third term on the right-hand size of \eqref{pfSketchEq6}, which measures the cross variation of this stochastic integral and $-F$, is equal to zero. This yields, for all $t$,
\begin{align*}
 d \O(\theta(t))  & \approx -dF(\theta(t)) \, \exp\int_0^t-\frac{\a_F(\theta(s))}{F(\theta(s))} + \a_F(\theta(t))\exp\int_0^t-\frac{\a_F(\theta(s))}{F(\theta(s))} \\
 & = \big(  \a_F(\theta(t)) - dF(\theta(t)) \big)\exp\int_0^t-\frac{\a_F(\theta(s))}{F(\theta(s))},
\end{align*}
which implies that
\begin{equation} \label{pfSketchEq7}
 \frac{d \O(\theta(t))}{\O(\theta(t))} \approx -\frac{\a_F(\theta(t))}{F(\theta(t))} + \frac{dF(\theta(t))}{F(\theta(t))},
\end{equation}
for all $t$. Since the right-hand sides of \eqref{pfSketchEq5} and \eqref{pfSketchEq7} are equivalent, it follows that
\begin{equation}
 \frac{V_s(t)}{V_m(t)} = \O(\theta(t)) = -F(\theta(t)) \exp\int_0^t\frac{\a_F(\theta(s))}{F(\theta(s))},
\end{equation}
for all $t$, which establishes \eqref{relReturnEq} and hence Theorem \ref{relValueThm}.

This derivation shows that, once \eqref{pfSketchEq3} is established, the proof of Theorem \ref{relValueThm} is simply a matter of applying \ito's lemma in a clever way. This proof sketch also highlights the manner in which our results rely on the continuous time framework \eqref{contSemimart}. \ito's lemma only holds for continuous time stochastic processes, and therefore the precision achieved by \eqref{relValueEq} requires the assumption that time is continuous. In the absence of a continuous time framework, the second-order Taylor approximations in the above proof sketch would be approximations only.

\subsection{Examples} \label{sec:examples}
Theorem \ref{relValueThm} is quite general and characterizes the performance of a broad class of portfolio strategies relative to the market. We wish to apply this general characterization to the two measures of price dispersion introduced in Section \ref{sec:distribution}, minus the geometric mean, $-G$, and minus the CES function, $-U$. The two corollaries that follow are simple applications of Theorem \ref{relValueThm}

\begin{cor} \label{returnsCor1}
The portfolio strategy $g(t) = (g_1(t), \ldots, g_N(t))$ with
\begin{equation} \label{strategyGEq}
 g_i(t) = \frac{V_g(t)}{N\theta_i(t)V_m(t)},
\end{equation}
for each $i = 1, \ldots, N$, has a value process $V_g$ that satisfies
\begin{equation} \label{relValueGEq}
 \log V_g(T) - \log V_m(T) = - \int_0^T\frac{\a_G(\theta(t))}{G(\theta(t))} + \log G(\theta(T)),
\end{equation}
for all $T$.
\end{cor}

In Corollary \ref{returnsCor1}, the shares of each asset $i$ held at time $t$, denoted by $g(t) = (g_1(t), \ldots, g_N(t))$, are calculated by evaluating \eqref{strategyEq} from Theorem \ref{relValueThm} using the measure of price dispersion $-G$. The results of this evaluation are given by \eqref{strategyGEq}. In terms of the portfolio weights $w^g$ defined in \eqref{weightsEq}, the shares $g_i$ imply an equal-weighted portfolio in which equal dollar amounts are invested in each asset, since
\begin{equation} \label{equalWeightsEq}
w^g_i(t) = \frac{g_i(t)p_i(t)}{V_g(t)} = \frac{1}{N},
\end{equation}
for $i = 1, \ldots, N$ and all $t$. For this reason, we shall refer to the portfolio strategy $g$ as the \emph{equal-weighted portfolio strategy}. One consequence of Theorem \ref{relValueThm} and Corollary \ref{returnsCor1}, then, is that the return of the equal-weighted strategy relative to the market can be decomposed into the non-negative drift $\a_G$ and changes in price dispersion as measured by minus the geometric mean of the asset price distribution, according to \eqref{relValueGEq}.

\begin{cor} \label{returnsCor2}
The portfolio strategy $u(t) = (u_1(t), \ldots, u_N(t))$ with
\begin{equation} \label{strategyUEq}
 u_i(t) = \frac{V_u(t)\theta^{\g-1}_i(t)}{V_m(t)U^{\g}(\theta(t))},
\end{equation}
for each $i = 1, \ldots, N$, has a value process $V_u$ that satisfies
\begin{equation} \label{relValueUEq}
 \log V_u(T) - \log V_m(T) = -\int_0^T\frac{\a_U(\theta(t))}{U(\theta(t))} + \log U(\theta(T)),
\end{equation}
for all $T$.
\end{cor}

 As with Corollary \ref{returnsCor1}, the shares of each asset $i$ held at time $t$ in Corollary \ref{returnsCor2}, denoted by $u(t) = (u_1(t), \ldots, u_N(t))$, are calculated by evaluating \eqref{strategyEq} using the measure of price dispersion $-U$ and the results of this evaluation are given by \eqref{strategyUEq}. The portfolio weights for the strategy $u$ are given by
\begin{equation} \label{cesWeightsEq}
 w^u_i(t) = \frac{u_i(t)p_i(t)}{V_u(t)} = \frac{p_i(t)\theta^{\g-1}_i(t)}{V_m(t)U^{\g}(\theta(t))} = \frac{\theta^{\g}_i(t)}{U^{\g}(\theta(t))} = \frac{\theta^{\g}_i(t)}{\sum_{j=1}^N \theta^{\g}_j(t)},
\end{equation}
for $i = 1, \ldots, N$ and all $t$. We shall refer to this portfolio strategy as the \emph{CES-weighted portfolio strategy}. Like with the equal-weighted strategy, Theorem \ref{relValueThm} and Corollary \ref{returnsCor2} imply that the return of the CES-weighted strategy relative to the market can be decomposed into the non-negative drift $\a_U$ and changes in price dispersion as measured by minus the CES function applied to the asset price distribution, according to \eqref{relValueUEq}.

Each different value of the non-negative CES parameter $\g$ implies a different CES function and hence a different portfolio strategy $u$. Note that as $\g$ tends to zero, the CES-weighted portfolio strategy converges to the equal-weighted strategy since the weights \eqref{cesWeightsEq} tend to $1/N$. For positive (negative) values of $\g$, the CES-weighted portfolio strategy is more (less) invested in higher-priced assets than the equal-weighted portfolio. Finally, if $\g$ is equal to one, then the CES-weighted portfolio is equivalent to the market portfolio, since in this case, for all $t$,
\begin{equation}
 u_i(t) = \frac{V_u(t)}{V_m(t)U(\theta(t))} = \frac{V_u(t)}{V_m(t)},
\end{equation}
and any portfolio strategy that purchases an equal number of shares of each asset is equivalent to the market portfolio.




\vskip 50pt

\section{Empirical Results} \label{sec:empirics}

Having characterized the relationship between relative returns and asset price distributions in full generality in Section \ref{sec:theory}, we now turn to an empirical analysis. We wish to investigate the accuracy of the decomposition in Theorem \ref{relValueThm} using real asset price data. In particular, we show in this section that the decompositions characterized in Corollaries \ref{returnsCor1} and \ref{returnsCor2} provide accurate descriptions of actual relative returns for the equal- and constant-elasticity-of-substitution (CES)-weighted portfolio strategies, as predicted by the theory.

\subsection{Data} \label{sec:data}
We use data on the prices of 30 different commodity futures from 1969-2018 to test our theoretical predictions. The choice to focus on commodity futures is motivated by two main factors. First, the two most important assumptions we impose on our theoretical framework --- that assets do not pay dividends, and that the market is closed and there is no asset entry or exit over time --- align fairly closely with commodity futures markets. These assets do not pay dividends, with returns driven entirely by capital gains. Commodity futures also rarely exit from the market, which is notable since such exit can substantially affect the relative returns of the equal- and CES-weighted portfolio strategies. In fact, no commodity futures contracts that we are aware of disappear from the market from 1969-2018, so this potential issue is irrelevant over the time period we consider. While new commodity futures contracts do enter into our data set between 1969-2018, such entry does not affect our empirical results and is easily incorporated into our framework as we explain in detail below.

Second, previous studies have already examined the return of equal- and CES-weighted portfolio strategies relative to the market using equities, so our choice of commodity futures provides an environment for truly novel empirical results. \citet{Vervuurt/Karatzas:2015}, for example, construct a CES-weighted portfolio of equities similar to the portfolio we construct for commodity futures below. These authors show that the CES-weighted equity portfolio consistently outperforms the market from 1990-2014 as predicted by Theorem \ref{relValueThm}, despite the fact that dividends and entry and exit in the form of IPOs and bankruptcies are important factors in equity markets.

Table \ref{commInfoTab} lists the start date and trading market for the 30 commodity futures in our 1969-2018 data set. These commodities encompass the four primary commodity domains (energy, metals, agriculture, and livestock) and span many bull and bear regimes. The table also reports the annualized average and standard deviation of daily log price changes over the lifetime of each futures contract. These data were obtained from the Pinnacle Data Corp., and report the two-month-ahead futures price of each commodity on each day that trading occurs, with the contracts rolled over each month.

Relative asset prices as defined by the $\theta_i$'s in \eqref{relPricesEq} are crucial to our theoretical framework and results. This concept, however, is essentially meaningless in the context of commodity prices, since different commodities are measured using different units such as barrels, bushels, and ounces. In order to give relative prices meaning in the context of commodity futures, we normalize all contracts with data on the January 2, 1969 start date so that their prices are equal to each other. All subsequent price changes occur without modification, meaning that price dynamics are unaffected by our normalization. For those commodities that enter into our data set after 1969, we set their initial log prices equal to the average log price of those commodities already in our data set on that date. After these commodities enter into the data set with a normalized price, all subsequent price changes occur without modification. The normalized commodity futures prices we construct are similar to price indexes, with all indexes set equal to each other on the initial start date and any indexes that enter after this start date set equal to the average of the existing indexes.

Figure \ref{relPricesFig} plots the normalized log commodity futures prices relative to the average for all 30 contracts in our data set from 1969-2018. This figure shows how normalized prices quickly disperse after the initial start date, with commodity futures prices constantly being affected by different shocks. After an initial period of rapid dispersion, however, the normalized commodity futures prices are roughly stable relative to each other with what looks like only modest increases in dispersion occurring after approximately 1980. These patterns are quantified and analyzed in our empirical analysis below.

\subsection{Portfolio Construction}
For our empirical analysis, it is necessary to construct a market portfolio strategy as defined by the weights \eqref{marketWeights}. In the context of commodity futures, the market portfolio cannot hold one share of each asset since futures contracts are simply agreements between two parties with no underlying asset held. This issue is easily resolved, however, since the market portfolio weights \eqref{marketWeights} are well-defined in the context of normalized commodity futures prices. In particular, \eqref{marketWeights} implies that the market portfolio invests in each commodity futures contract an amount that is proportional to the normalized price of that commodity. For this reason, we often refer to the market portfolio strategy as the price-weighted market portfolio strategy in the empirical analysis of this section. Note that the market portfolio of commodity futures requires no rebalancing, since price changes automatically cause the weights of each commodity in the portfolio to change in a manner that is consistent with price-weighting.

In addition to the price-weighted market portfolio, we construct equal- and CES-weighted portfolios of commodity futures as described in Corollaries \ref{returnsCor1} and \ref{returnsCor2}. The weights that define these two portfolio strategies are given by \eqref{equalWeightsEq} and \eqref{cesWeightsEq}, and are constructed using the normalized prices for which relative price is a meaningful concept. For the CES-weighted portfolio strategy, we set the value of $\g$ equal to $-0.5$, meaning that this portfolio places greater weight on lower-priced commodity futures than does the equal-weighted portfolio (see the discussion at the end of Section \ref{sec:examples}). Both the equal- and CES-weighted portfolio strategies require active rebalancing since, unlike the price-weighted market portfolio, their weights tend to deviate from \eqref{equalWeightsEq} and \eqref{cesWeightsEq} as prices change over time. Each portfolio is rebalanced once each month. Finally, even though our commodity futures data cover 1969-2018, the fact that we normalize prices by setting them equal to each other on the 1969 start date implies that the distribution of relative prices will have little meaning until these prices are given time to disburse. In a manner similar to the commodity value measure of \citet{Asness/Moskowitz/Pedersen:2013}, we wait five years before forming the equal-, CES-, and price-weighted market portfolios, so that these portfolios are constructed using normalized prices from 1974-2018.

\subsection{Results} \label{sec:results}
Figure \ref{returnsFig} plots the log cumulative returns for the price-weighted (market) portfolio strategy and the equal- and CES-weighted portfolio strategies from 1974-2018. The figure shows that all three portfolios have roughly similar behavior over time, but that the monthly rebalanced equal- and CES-weighted portfolios gradually and consistently outperform the price-weighted portfolio over time. These patterns are quantified in Table \ref{returnsTab}, which reports the annualized average and standard deviation of monthly returns for all three portfolio strategies over this time period. The monthly returns of the market portfolio have correlations of 0.95 and 0.89 with the returns of the equal- and CES-weighted portfolios, respectively. The outperformance of the equal- and CES-weighted portfolio strategies relative to the price-weighted market portfolio is also evident in Table \ref{relReturnsTab}, which reports the annualized average, standard deviation, and Sharpe ratio of monthly relative returns for the equal- and CES-weighted portfolios from 1974-2018. Tables \ref{returnsTab} and \ref{relReturnsTab} also report returns statistics for each decade in our long sample period.


The results of Tables \ref{returnsTab} and \ref{relReturnsTab} show that the equal- and CES-weighted portfolios consistently and substantially outperformed the price-weighted market portfolio over the 1974-2018 time period. This  outperformance is most evident from the high Sharpe ratios for the excess returns of both the equal- and CES-weighted portfolio, as shown in Table \ref{relReturnsTab}. Notably, both of these Sharpe ratios consistently rise above 0.5 after 1980, which is after most of the commodity futures contracts in our data set have started trading according to Table \ref{commInfoTab}. In other words, as the number of tradable assets $N$ rises, portfolio outperformance also rises. This is not surprising, since a greater number of tradable assets generally implies a greater value for the non-negative drift process $\a_F$, and it is this process that mostly determines relative portfolio returns over long horizons, as we demonstrate below.

The general theory of Section \ref{sec:theory} does not make any statements about the size of portfolio returns. Instead, this theory states that the returns for a large class of portfolio strategies relative to the market can be decomposed into a non-negative drift and changes in asset price dispersion, according to Theorem \ref{relValueThm}. When applied to the equal- and CES-weighted portfolios as in Corollaries \ref{returnsCor1} and \ref{returnsCor2}, this implies that the relative return of the equal-weighted portfolio strategy can be decomposed into the drift process adjusted by price dispersion, $-\a_G/G \geq 0$, and changes in the geometric mean of the asset price distribution, as in \eqref{relValueGEq}. Similarly, the relative return of the CES-weighted portfolio strategy can be decomposed into the drift process adjusted by price dispersion, $-\a_U/U \geq 0$, and changes in the CES function applied to the asset price distribution, as in \eqref{relValueUEq}.

In order to empirically investigate the decomposition \eqref{relValueEq} from Theorem \ref{relValueThm}, in Figure \ref{returnsEWFig} we plot the cumulative abnormal returns --- returns relative to the price-weighted market portfolio strategy --- of the equal-weighted portfolio strategy together with the cumulative value of the drift process adjusted by price dispersion, $-\a_G/G$, from 1974-2018. In addition, Figure \ref{dispEWFig} plots price dispersion as measured by minus the log of the geometric mean of the commodity price distribution, $G$, normalized relative to its average value for 1974-2018. In addition to the consistent and substantial outperformance of the equal-weighted portfolio relative to the price-weighted portfolio, these figures show that short-run relative return fluctuations for the equal-weighted portfolio closely follow fluctuations in commodity price dispersion while the long-run behavior of these relative returns closely follow the smooth adjusted drift. Indeed, there is a striking contrast between the high volatility of price dispersion in Figure \ref{dispEWFig} and the near-zero volatility of the adjusted drift in Figure \ref{returnsEWFig}. This is an important observation that is a direct prediction of Theorem \ref{relValueThm} and Corollary \ref{returnsCor1}, a point we discuss further below.

In addition to the contrasting volatilities of price dispersion and the adjusted drift, Figures \ref{returnsEWFig} and \ref{dispEWFig} show that the cumulative abnormal returns of the equal-weighted portfolio strategy are equal to the cumulative value of the adjusted drift process, $\int_0^T -\a_G(\theta(t))/G(\theta(t))$, plus the log of the geometric mean of the commodity price distribution, $\log G(\theta(T))$. Indeed, the solid black line in Figure \ref{returnsEWFig} (cumulative abnormal returns) is equal to the dashed red line in that same figure (cumulative value of the adjusted drift process) minus the line in Figure \ref{dispEWFig} (minus the log of the geometric mean of the commodity price distribution). This is exactly the relationship described by \eqref{relValueGEq} from Corollary \ref{returnsCor1}. We stress, however, that this empirical relationship is a necessary consequence of how the non-negative adjusted drift process, $-\a_G/G$, is calculated. For each day that we have data, the cumulative value of $-\a_G/G$ up to that day is calculated by subtracting the log value of the geometric mean of the commodity price distribution, $\log G$, from the cumulative abnormal returns, $\log V_g - \log V_m$, according to the identity \eqref{relValueGEq} from Corollary \ref{returnsCor1}.

Given that the empirical decomposition of Figures \ref{returnsEWFig} and \ref{dispEWFig} is constructed so that \eqref{relValueGEq} must hold, it is natural to wonder what the usefulness of this decomposition is. Some of this usefulness lies in the prediction that one part of this decomposition, the cumulative value of the adjusted drift process, $-\a_G/G$, is non-decreasing. This prediction is clearly confirmed by the smooth upward slope of the cumulative value of the adjusted drift line in Figure \ref{returnsEWFig}, and has implications for the long-run relative performance of the equal- and price-weighted portfolio strategies, as we discuss below. Most of the usefulness of the decomposition \eqref{relValueGEq} lies, however, in the prediction that the cumulative value of the adjusted drift process is a finite variation process, while the other part, the log value of the geometric mean of the commodity price distribution, $\log G$, is not. Recall from the discussions in Sections \ref{sec:setup} and \ref{sec:generalPort} that a finite variation process has zero quadratic variation, or zero instantaneous variance. To be clear, the prediction that the cumulative value of the adjusted drift process is a finite variation process is not a prediction that the sample variance of changes in the cumulative value of the adjusted drift process computed using monthly, discrete-time data will be equal to zero, but rather a prediction that these changes will be roughly constant over time.\footnote{Note that the sample variance of a continuous-time finite variation process computed using discrete-time data will never be exactly equal to zero.} In other words, our results predict that the cumulative value of the adjusted drift process will grow at a roughly constant rate with only few and small changes over time.

This smooth growth is exactly what is observed in the dashed red line of Figure \ref{returnsEWFig}, and, as mentioned above, is in stark contrast to the highly volatile behavior of price dispersion shown in Figure \ref{dispEWFig}. This contrast can be quantified by noting that the coefficient of variation of changes in the cumulative value of the adjusted drift process is equal to 3.14, while the coefficient of variation of changes in the log of minus price dispersion is equal to 124.64. These results confirm one of the key predictions of Theorem \ref{relValueThm} and Corollary \ref{returnsCor1}.

The positive and relatively constant values of the adjusted drift, $-\a_G/G$, over time have an important implication for the long-run return of the equal-weighted portfolio strategy relative to the price-weighted market portfolio strategy. Since \eqref{relReturnEq} and Theorem \ref{relValueThm} imply that relative returns can be decomposed into the adjusted drift and changes in asset price dispersion, a consistently positive adjusted drift over long time horizons can only be counterbalanced by consistently rising asset price dispersion. In the absence of such rising dispersion, the positive drift guarantees outperformance relative to the market. Therefore, the relatively small increase in commodity price dispersion shown in Figure \ref{dispEWFig} together with the positive values of the adjusted drift shown in Figure \ref{returnsEWFig} ensure that the equal-weighted portfolio outperforms the market portfolio over the 1974-2018 time period.

In a similar manner to Figure \ref{returnsEWFig}, Figure \ref{returnsCESFig} plots the cumulative abnormal returns of the CES-weighted portfolio strategy together with the cumulative value of the drift process adjusted by price dispersion, $-\a_U/U$, over the same 1974-2018 time period. Figure \ref{dispCESFig} plots price dispersion as measured by minus the log of the CES function applied to the asset price distribution, $U$, normalized relative to its average value over this time period. As with the equal-weighted portfolio, the cumulative value of the adjusted drift process in Figure \ref{returnsCESFig} is calculated using the identity \eqref{relValueUEq} from Corollary \ref{returnsCor2}. The results in Figures \ref{returnsCESFig} and \ref{dispCESFig} for the CES-weighted portfolio align closely with the results in Figures \ref{returnsEWFig} and \ref{dispEWFig} for the equal-weighted portfolio.

Indeed, Figures \ref{returnsCESFig} and \ref{dispCESFig} show that short-run relative return fluctuations for the CES-weighted portfolio strategy closely follow fluctuations in commodity price dispersion, as measured by minus the CES function, while the long-run behavior of these relative returns closely follow the smoothly accumulating adjusted drift. Much like in Figure \ref{returnsEWFig}, Figure \ref{returnsCESFig} shows that the cumulative value of the adjusted drift, which is a finite variation process according to Theorem \ref{relValueThm}, grows at a roughly constant rate over time, with a clear contrast between this stable growth and the rapid fluctuations in price dispersion shown in Figure \ref{dispCESFig}. As discussed above, the fact that the adjusted drift, $-\a_U/U$, is approximately constant over time is consistent with the prediction that its cumulative value is a finite variation process, thus confirming one of the key results in Theorem \ref{relValueThm} and Corollary \ref{returnsCor2}. Finally, Figure \ref{returnsCESFig} confirms the consistent and substantial outperformance of the CES-weighted portfolio relative to the price-weighted market portfolio like in Tables \ref{returnsTab} and \ref{relReturnsTab}. As with the equal-weighted portfolio, this long-run outperformance is predicted by \eqref{relValueEq} and Theorem \ref{relValueThm} given the relatively small change in price dispersion observed in Figure \ref{dispCESFig} compared to the large increase in the cumulative value of the adjusted drift observed in Figure \ref{returnsCESFig}.


\vskip 50pt

\section{Discussion} \label{sec:discuss}
The empirical results shown in Figures \ref{returnsEWFig}-\ref{dispCESFig} confirm the prediction of Theorem \ref{relValueThm} and Corollaries \ref{returnsCor1} and \ref{returnsCor2} that the drift component of the decomposition \eqref{relReturnEq} is nearly constant. As a consequence, this decomposition and its intuitive version \eqref{intuitiveEq} can be understood as
\begin{equation} \label{intuitiveEq2}
 \text{relative return} \;\; = \;\; \text{constant} \; - \; \text{change in asset price dispersion}.
\end{equation}
Furthermore, the results of Figures \ref{returnsEWFig} and \ref{returnsCESFig} clearly show that this non-negative constant drift is in fact positive in the case of the equal- and CES-weighted portfolios of commodity futures.


\subsection{The Price Dispersion Asset Pricing Factor}
Taken together, our theoretical and empirical results show that changes in asset price dispersion are key determinants of the returns for a large class of portfolios relative to the market. Thus, the distribution of relative asset prices, as measured by the dispersion of those prices, is necessarily an asset pricing factor. This fact is apparent from \eqref{intuitiveEq2}, which is setup in the same way as empirical asset pricing factor regression models \citep{Fama/French:1993}. Crucially, however, the theoretical results of Theorem \ref{relValueThm} that establish the intuitive version \eqref{intuitiveEq2} are achieved under minimal assumptions that should be consistent with essentially any model of asset pricing, meaning that this price dispersion factor is universal across different economic and financial environments. Our empirical results in Figures \ref{returnsEWFig}-\ref{dispCESFig} help to confirm this universality, especially when taken together with previous studies documenting the accuracy of the decomposition of Theorem \ref{relValueThm} for U.S.\ equity markets \citep{Vervuurt/Karatzas:2015}.

The generality of our results provides a novel workaround to many of the criticisms that have been raised recently about the empirical asset pricing literature. In particular, the implausibly high and rising number of factors and anomalies that this literature has identified has drawn a number of rebukes. \citet{Harvey/Liu/Zhu:2016}, for example, examine hundreds of different asset pricing factors and anomalies that have been uncovered using standard empirical methods and conclude that most are likely invalid. They also propose a substantially higher standard for statistical significance in future empirical analyses. Similarly, \citet{Bryzgalova:2016} shows that standard empirical methods applied to inappropriate risk factors in linear asset pricing models can generate spuriously high significance. \citet{Novy-Marx:2014} provides a different critique, demonstrating that many supposedly different anomalies are potentially driven by one or two common risk factors. All of these studies suggest that the extensive list of factors and anomalies proposed by the literature overstates the true number. The asset price dispersion factor established by Theorem \ref{relValueThm} is not derived using a specific economic model or a specific regression framework, but rather using general mathematical methods that represent asset prices as continuous semimartingales that are consistent with essentially all models and empirical specifications. For this reason, the price dispersion asset pricing factor we characterize is not subject to the criticisms of this literature.


\subsection{Price Dispersion, Value, and Naive Diversification}
The results of Theorem \ref{relValueThm} and Corollaries \ref{returnsCor1} and \ref{returnsCor2} offer new interpretations for both the value anomaly for commodities uncovered by \citet{Asness/Moskowitz/Pedersen:2013} and the surprising effectiveness of naive $1/N$ diversification described by \citet{DeMiguel/Garlappi/Uppal:2009}. The value anomaly for commodities of \citet{Asness/Moskowitz/Pedersen:2013} is constructed by ranking commodity futures prices relative to their average price five years earlier and then comparing the returns of portfolios of low-rank, high-value commodities to the returns of portfolios of high-rank, low-value commodities. This price ranking system is similar to the price normalization we implement based on the prices of commodity futures on the 1969 start date. Because the equal- and CES-weighted commodity futures portfolios put more weight on lower-normalized-priced commodities than does the price-weighted market portfolio, it follows that the predictable excess returns that we report in Tables \ref{returnsTab} and \ref{relReturnsTab} are similar to the value effect for commodities of \citet{Asness/Moskowitz/Pedersen:2013}.

A key difference between these results and our results, however, is that we link the predictable excess returns of the equal- and CES-weighted portfolio strategies to the approximate stability of commodity price dispersion as measured by minus the geometric mean and CES functions. This link is essential to understanding the economic and financial mechanisms behind these excess returns. Our results imply that such excess returns decrease as asset price dispersion rises, with positive excess returns ensured only if asset price dispersion does not rise substantially over time. Thus, any attempt to explain the predictable excess returns of Tables \ref{returnsTab} and \ref{relReturnsTab} and the related value anomaly for commodity futures must also explain the fluctuations in commodity price dispersion shown in Figures \ref{dispEWFig} and \ref{dispCESFig}, since these fluctuations are driving the excess return fluctuations. This conclusion points to the importance of a deeper understanding of the economic and financial mechanisms behind fluctuations in asset price dispersion.

\citet{DeMiguel/Garlappi/Uppal:2009} consider a number of different portfolio diversification strategies using several different data sets and show that a naive strategy of weighting each asset equally consistently outperforms almost all of the more sophisticated strategies. One of the strategies that is outperformed by naive $1/N$ diversification is a value-weighted market portfolio strategy based on CAPM. This strategy is equivalent to the market portfolio we defined in Section \ref{sec:port}, since our price weights are equivalent to their value weights. The results of Corollary \ref{returnsCor1}, therefore, can be applied to the excess returns of the equal-weighted portfolio relative to the value-weighted market portfolio uncovered by \citet{DeMiguel/Garlappi/Uppal:2009}. In particular, the corollary implies that this excess return is determined by a non-negative drift and the change in asset price dispersion as measured by minus the geometric mean of relative prices.

The decomposition of Corollary \ref{returnsCor1} provides a novel interpretation of the results of \citet{DeMiguel/Garlappi/Uppal:2009} in terms of the stability of the distributions of the various different empirical data sets these authors consider. As with the value anomaly for commodities, our theoretical decomposition implies that the excess returns of the naive $1/N$ diversification strategy decrease as asset price dispersion rises and are positive only if dispersion does not substantially rise over time. Thus, our results strongly suggest that the values of the various different assets considered by \citet{DeMiguel/Garlappi/Uppal:2009} are stable relative to each other, in a manner similar to what we observe for commodity futures prices in Figures \ref{dispEWFig} and \ref{dispCESFig}. In the absence of such stability, there would be no reason to expect the equal-weighted portfolio to outperform the value-weighted market portfolio as the authors observe. Once again, these conclusions highlight the importance of a deeper understanding of the economic and financial mechanisms behind fluctuations in asset price dispersion.

\subsection{Price Dispersion and Efficient Markets}
The relative return decomposition of Theorem \ref{relValueThm} reveals a novel dichotomy for markets in which dividends and asset entry/exit over time play small roles. On the one hand, the dispersion of asset prices may be approximately stable over time, in which case \eqref{intuitiveEq2} implies that predictable excess returns exist for a large class of portfolio strategies. This is the scenario we observe for commodity futures in Figures \ref{returnsEWFig}-\ref{dispCESFig}. In such markets, fluctuations in asset price dispersion are linked to excess returns via the accounting identity \eqref{relReturnEq} from Theorem \ref{relValueThm}. In a standard equilibrium model of asset pricing, these predictable excess returns may exist only if they are compensation for risk. This risk, in turn, is defined by an endogenous stochastic discount factor that is linked to the marginal utility of economic agents. It is not clear, however, how marginal utility might be linked to the dispersion of asset prices. It is also not clear why marginal utility should be higher when asset prices grow more dispersed, yet these are necessary implications of any standard asset pricing model in which price dispersion is asymptotically stable, according to our results.

On the other hand, the dispersion of asset prices may not be stable over time. In this case, asset price dispersion is consistently and rapidly rising, and the decomposition \eqref{intuitiveEq2} no longer predicts excess returns. Instead, this decomposition predicts rising price dispersion that cancels out the non-negative drift component of \eqref{intuitiveEq2} on average over time. The relative return decomposition of Theorem \ref{relValueThm} makes no predictions about the stability of asset price dispersion, so this possibility is not ruled out by our theoretical results. Nonetheless, it is notable that both our empirical results for commodity futures and the empirical results of \citet{Vervuurt/Karatzas:2015} for U.S.\ equities are inconsistent with this no-stability, no-excess-returns market structure. In light of these results, future work that examines the long-run properties of price dispersion in different asset markets and attempts to distinguish between the two sides of this dichotomy --- asymptotically stable markets with predictable excess returns versus asymptotically unstable markets without predictable excess returns --- may yield interesting new insights.

This novel dichotomy has several implications. First, it provides a new interpretation of market efficiency in terms of a constraint on cross-sectional asset price dynamics and the dispersion of relative asset prices. Either asset price dispersion rises consistently and rapidly over time, consistent with this constraint, or there exists a market inefficiency or a risk factor based on the decomposition \eqref{relReturnEq}. Second, it raises the possibility that well-known asset pricing risk factors such as value, momentum, and size \citep{Banz:1981,Fama/French:1993,Asness/Moskowitz/Pedersen:2013} may be interpretable in terms of the dynamics of asset price dispersion. To the extent that the decomposition of Theorem \eqref{relValueThm} is universal, the predictable excess returns underlying each of these risk factors may potentially be linked to a violation of the constraint on cross-sectional asset price dynamics and the dispersion of relative asset prices mentioned above. In other words, traditional asset pricing risk factors imply specific behavior for asset price dispersion over time and hence may be interpreted in terms of that specific behavior.

\vskip 50pt

\section{Conclusion} \label{sec:conclusion}

We represent asset prices as general continuous semimartingales and show that the returns on a large class of portfolio strategies relative to the market can be decomposed into a non-negative drift and changes in asset price dispersion. Because of the minimal assumptions underlying this result, our decomposition is little more than an accounting identity that is consistent with essentially any asset pricing model. We show that the drift component of our decomposition is approximately constant over time, thus implying that changes in asset price dispersion determine relative return fluctuations. This conclusion reveals an asset pricing factor --- changes in asset price dispersion --- that is universal across different economic and financial environments. We confirm our theoretical predictions using commodity futures, and show that equal- and constant-elasticity-of-substitution-weighted portfolios consistently and substantially outperformed the price-weighted market portfolio from 1974-2018.

\vskip 50pt

\begin{spacing}{1.2}

\bibliographystyle{chicago}
\bibliography{econ}

\end{spacing}

\vskip 50pt

\begin{spacing}{1.1}

\appendix

\section{Proofs} \label{proofs}

This appendix presents the proof of Lemmas \ref{dispersionLem} and \ref{alphaLem}, and Theorem \ref{relValueThm}.

\begin{proofLemma1}
Let $F$ be a measure of price dispersion. Suppose that $\theta, \theta' \in \D$, with
\begin{equation*}
 \max(\theta) = \max(\theta_1, \ldots, \theta_N) > \max(\theta'_1, \ldots, \theta'_N) = \max(\theta'),
\end{equation*}
and that $\theta_i = \theta'_i$ for all $i$ in some subset of $\{1, \ldots, N\}$ that contains $N - 2$ elements. Without loss of generality, we assume that $\max(\theta) = \theta_1$, $\max(\theta') = \theta'_1$, and $\theta_2 \neq \theta'_2$. Note that this implies that $\theta_1 - \theta'_1 = \theta'_2 - \theta_2 > 0$, since both $\theta$ and $\theta'$ must add up to one.

Let $\tilde{\theta} = (\theta_2, \theta_1, \theta_3, \ldots, \theta_N)$ be the relative price vector obtained by exchanging the first two elements of $\theta$. Because a measure of price dispersion is invariant to permutations of $\theta$ by definition, it follows that $F(\theta) = F(\tilde{\theta})$. Let
\begin{equation*}
 \b = \frac{\theta'_1 - \theta_2}{\theta_1 - \theta_2},
\end{equation*}
and note that $0 < \b < 1$ and
\begin{equation*}
 \b\theta + (1 - \b)\tilde{\theta} = \theta'.
\end{equation*}
Because $F$ is convex by definition, we have that
\begin{equation} \label{ineqApp}
 F(\theta) = \b F(\theta) + (1 - \b)F(\tilde{\theta}) \geq F(\b\theta + (1 - \b)\tilde{\theta}) = F(\theta').
\end{equation}
If $F$ is strictly convex, then the inequality in \eqref{ineqApp} becomes a strict inequality.
\end{proofLemma1}

\begin{proofLemma2}
For any continuous semimartingale vector $z$ with $z(t) \in \R^N$ for all $t$, let $d\Var(z)$ denote the covariance matrix $(d \langle z_i, z_j \rangle )_{1 \leq i, j \leq N}$. Any measure of price dispersion $F$ is convex by definition, so it follows that the Hessian matrix $H_F$ is positive semidefinite. For a given $t$, this implies that $H_F$ has eigenvalues $\l_1, \ldots, \l_N \geq 0$, with corresponding eigenvectors $e_k = (e_{k1}, \ldots, e_{kN})$, $k = 1, \ldots, N$, such that
\begin{equation} \label{alphaPfEq1}
 F_{ij}(\theta(t)) = \sum_{k=1}^N \l_k e_{ki}e_{kj},
\end{equation}
for $i, j = 1, \ldots, N$. Letting $x^T$ denote the transpose of a vector $x \in \R^N$, it follows that
\begin{equation*}
  \sum_{i, j = 1}^N F_{ij}(\theta(t)) \, d \langle \theta_i, \theta_j \rangle (t) = \sum_{k=1}^N \l_k \sum_{i, j = 1}^N e_{ki}e_{kj} \, d \langle \theta_i, \theta_j \rangle (t) =  \sum_{k=1}^N \l_k e_k \, d\Var(\theta)(t) \, e^T_k \geq 0,
\end{equation*}
for all $t$, since the covariance matrix $d\Var(\theta)$ is positive semidefinite. Of course, this implies that $\a_F \geq 0$ as well.

Now suppose that $\operatorname{rank}(F) > 1$ and that $d\Var(p)(t)$ is positive definite, for all $t$. Note that
\begin{equation*}
 d\Var(\log p)(t) = p(t)\,d\Var(p)(t)\,p^T(t),
\end{equation*}
for all $t$, so that $d\Var(\log p)$ is positive definite if $d\Var(p)$ is positive definite. Furthermore, \citet{Fernholz:2002} shows that if $d\Var(\log p)$ is positive definite, then $d\Var(\log \theta)$ is positive semidefinite with null space generated by $\theta$. According to \eqref{alphaPfEq1},
\begin{align*}
  \sum_{i, j = 1}^N F_{ij}(\theta(t)) \, d \langle \theta_i, \theta_j \rangle (t) & = \sum_{k=1}^N \l_k \sum_{i, j = 1}^N e_{ki}e_{kj}\theta_i(t)\theta_j(t) \, d \langle \log\theta_i, \log\theta_j \rangle (t) \\
  & =  \sum_{k=1}^N \l_k e_k\theta(t) \, d\Var(\log \theta)(t) \, \theta^T(t)e^T_k,
\end{align*}
for all $t$. We know that at least two of the eigenvalues $\l_1, \ldots, \l_N$ are positive since $\operatorname{rank}(F) > 1$ and $H_F$ is positive semidefinite, so we assume, without loss of generality, that $\l_1, \l_2 > 0$. It follows, then, that
 \begin{equation*}
  \sum_{i, j = 1}^N F_{ij}(\theta(t)) \, d \langle \theta_i, \theta_j \rangle (t) = \sum_{k=1}^N \l_k e_k\theta(t) \, d\Var(\log \theta)(t) \, \theta^T(t)e^T_k > 0,
 \end{equation*}
for all $t$, since
\begin{equation*}
 e_k\theta(t) \, d\Var(\log \theta)(t) \, \theta^T(t)e^T_k > 0,
\end{equation*}
for either $k = 1$ or $k = 2$, for all $t$. This implies that $\a_F > 0$.
\end{proofLemma2}

\begin{proofThm1}
Theorem \ref{relValueThm} follows from the more general results in Proposition 4.8 of \citet{Karatzas/Ruf:2017}. To see this, let $\tilde{F} = -F$, and note that by definition
\begin{equation*}
 \a_{\tilde{F}}(\theta(t)) = -\a_F(\theta(t)),
\end{equation*}
for all $t$. The function $\tilde{F}$ is regular according to Definition 3.1 of \citet{Karatzas/Ruf:2017} because it is continuous and concave and we have assumed that prices are always positive. Furthermore, because $\tilde{F}$ is twice continuously differentiable, it follows that the finite variation process $\Gamma^{\tilde{F}}$, defined in (3.2) of \citet{Karatzas/Ruf:2017}, satisfies
\begin{equation*}
 d\Gamma^{\tilde{F}}(t) = -\a_{\tilde{F}}(\theta(t)) = \a_F(\theta(t)),
\end{equation*}
for all $t$. Proposition 4.8 then yields the result \eqref{relValueEq} of Theorem \ref{relValueThm}.

In addition to this proof, we wish to informally derive \eqref{pfSketchEq3}, which played an important role in the proof sketch of Theorem \ref{relValueThm} in Section \ref{sec:proof}. This equation states that
\begin{equation*}
 d\frac{V_s(t)}{V_m(t)} = \sum_{i=1}^Ns_i(t) \, d\theta_i(t),
\end{equation*}
for all $t$. For notational simplicity, we drop all time dependences in this informal derivation and simply write $f$ for $f(t)$, for any function of time $f$. For all $i = 1, \ldots, N$, the second-order Taylor approximation of $\log \theta_i$ is given by
\begin{equation} \label{pfAppendixEq1}
 d\log \theta_i \approx \frac{d\theta_i}{\theta_i} - \frac{1}{2}\frac{d^2\theta_i}{\theta^2_i},
\end{equation}
which implies that
\begin{align}
 \sum_{i=1}^Ns_i \, d\theta_i & \approx \sum_{i=1}^N s_i \left( \theta_i \, d\log \theta_i + \frac{1}{2}\frac{d^2\theta_i}{\theta_i} \right) \notag \\
  & = \sum_{i=1}^N \frac{V_s}{V_m} w^s_i \, d\log \theta_i + \sum_{i=1}^N \frac{V_s}{2V_m} w^s_i \frac{d^2\theta_i}{\theta^2_i}, \label{pfAppendixEq2}
\end{align}
where the last equality follows from the definition of portfolio strategy weights \eqref{weightsEq}.

According to \eqref{pfAppendixEq1}, for all $i, j = 1, \ldots, N$,
\begin{align}
 \frac{d\theta_i \, d\theta_j}{\theta_i\theta_j} & \approx \left( d\log \theta_i + \frac{1}{2}\frac{d^2\theta_i}{\theta^2_i} \right)\left( d\log \theta_j + \frac{1}{2}\frac{d^2\theta_j}{\theta^2_j} \right) \notag \\
  & = d\log \theta_i \, d\log \theta_j + \frac{1}{2}\frac{d^2\theta_i \, d\log \theta_j}{\theta^2_i} + \frac{1}{2}\frac{d^2\theta_j \, d\log \theta_i}{\theta^2_i} + \frac{1}{4}\frac{d^2\theta_i \, d^2\theta_j}{\theta^2_i\theta^2_j} \notag \\
  & \approx d\log \theta_i \, d\log \theta_j. \label{pfAppendixEq3}
\end{align}
In a manner consistent with our use of second-order Taylor approximations in this proof sketch, the bottom approximation \eqref{pfAppendixEq3} follows by assuming that all third-order terms or higher --- those terms with the differential operator $d$ raised to a power greater than or equal to three --- are approximately equal to zero. Of course, in continuous time \ito's lemma ensures that these higher-order terms are in fact equal to zero. If we substitute \eqref{pfAppendixEq3} into \eqref{pfAppendixEq2}, then we have
\begin{equation} \label{pfAppendixEq100}
  \sum_{i=1}^Ns_i \, d\theta_i \approx \sum_{i=1}^N \frac{V_s}{V_m} w^s_i \, d\log \theta_i + \sum_{i=1}^N \frac{V_s}{2V_m} w^s_i \, d^2\log \theta_i.
\end{equation}

The second-order Taylor approximation of $V_s/V_m$ is given by
\begin{equation*}
 d\frac{V_s}{V_m} \approx \frac{V_s}{V_m} \, d\log (V_s/V_m) + \frac{1}{2} \frac{d^2 (V_s/V_m)}{V_s/V_m},
\end{equation*}
which is equivalent to
\begin{equation} \label{pfAppendixEq4}
 \frac{d (V_s/V_m)}{V_s/V_m} \approx d\log (V_s/V_m) + \frac{1}{2} \frac{d^2 (V_s/V_m)}{(V_s/V_m)^2}.
\end{equation}
According to \eqref{weightsEq} and \eqref{selfFinanceEq},
\begin{equation*}
 \frac{dV_s}{V_s} \approx \sum_{i=1}^N s_i \, \frac{dp_i}{V_s} = \sum_{i=1}^N w^s_i \, \frac{dp_i}{p_i},
\end{equation*}
which implies that the second-order Taylor approximation of $\log V_s$ can be written as
\begin{align*}
 d \log V_s & \approx \frac{dV_s}{V_s} - \frac{1}{2}\frac{d^2 V_s}{V^2_s}  \\
 &  \approx \sum_{i=1}^N w^s_i \, \frac{dp_i}{p_i} - \frac{1}{2}\sum_{i,j = 1}^N w^s_iw^s_j \, \frac{dp_i \, dp_j}{p_ip_j} \\
 & \approx \sum_{i=1}^N w^s_i \, d\log p_i + \frac{1}{2}\sum_{i=1}^N w^s_i \, \frac{d^2p_i}{p^2_i} - \frac{1}{2}\sum_{i,j = 1}^N w^s_iw^s_j \, \frac{dp_i \, dp_j}{p_ip_j},
\end{align*}
where the last equality follows from \eqref{pfAppendixEq1}. It follows that
\begin{align}
 d \log (V_s/ V_m) & \approx \sum_{i=1}^N w^s_i \, d\log \theta_i + \frac{1}{2}\sum_{i=1}^N w^s_i \, \frac{d^2p_i}{p^2_i} - \frac{1}{2}\sum_{i,j = 1}^N w^s_iw^s_j \, \frac{dp_i \, dp_j}{p_ip_j} \\
  & \approx \sum_{i=1}^N w^s_i \, d\log \theta_i + \frac{1}{2}\sum_{i=1}^N w^s_i \, d^2\log p_i - \frac{1}{2}\sum_{i,j = 1}^N w^s_iw^s_j \, d\log p_i \, d\log p_j, \label{pfAppendixEq5}
\end{align}
where the last equality follows from \eqref{pfAppendixEq3}. Since $V_m = p_1 + \cdots + p_N$ by \eqref{marketValueEq}, we have that
\begin{align}
  \sum_{i, j = 1}^N w^s_iw^s_j \, d\log \theta_i \, d\log \theta_j & = \sum_{i, j = 1}^N w^s_iw^s_j(d\log p_i - d\log V_m(p))(d\log p_j - d\log V_m) \notag \\
  & =  \sum_{i,j = 1}^N w^s_iw^s_j \, d\log p_i \, d\log p_j - 2\sum_{i=1}^N w^s_i \,d \log p_i \, d\log V_m + d^2\log V_m. \label{pfAppendixEq6}
\end{align}
A similar argument to the one above proves that
\begin{equation} \label{pfAppendixEq7}
 \sum_{i = 1}^N w^s_i \, d^2\log \theta_i = \sum_{i = 1}^N w^s_i \, d^2\log p_i -  2\sum_{i=1}^N w^s_i \, d\log p_i \, d\log V_m + d^2\log V_m.
\end{equation}
Substituting \eqref{pfAppendixEq6} and \eqref{pfAppendixEq7} into \eqref{pfAppendixEq5} yields
\begin{equation} \label{pfAppendixEq8}
 d \log ( V_s/V_m ) \approx \sum_{i=1}^N w^s_i \, d\log \theta_i + \frac{1}{2}\sum_{i=1}^N w^s_i \, d^2\log \theta_i - \frac{1}{2}\sum_{i,j = 1}^N w^s_iw^s_j \, d\log \theta_i \, d\log \theta_j,
\end{equation}
since the last two terms of \eqref{pfAppendixEq6} and \eqref{pfAppendixEq7} cancel each other out.

Suppose that \eqref{pfSketchEq3} holds. In this case, we have
\begin{equation}
 \frac{d(V_s/V_m)}{V_s/V_m} = \sum_{i=1}^N \frac{s_iV_m}{V_s} \,d\theta_i =  \sum_{i=1}^N w^s_i \, \frac{d\theta_i}{\theta_i},
\end{equation}
and hence also
\begin{align}
  \frac{d^2(V_s/V_m)}{(V_s/V_m)^2} & = \left( \sum_{i=1}^N w^s_i \, \frac{d\theta_i}{\theta_i} \right) \left( \sum_{i=1}^N w^s_i \, \frac{d\theta_i}{\theta_i} \right) \notag \\
  & = \sum_{i,j = 1}^N w^s_iw^s_j \, \frac{d\theta_i \, d\theta_j}{\theta_i \theta_j} \notag \\
  & \approx \sum_{i,j = 1}^N w^s_iw^s_j \, d\log \theta_i \, d\log \theta_j. \label{pfAppendixEq9}
\end{align}
If we substitute \eqref{pfAppendixEq8} and \eqref{pfAppendixEq9} into \eqref{pfAppendixEq4}, we have
\begin{equation} \label{pfAppendixEq10}
 \frac{d(V_s/V_m)}{V_s/V_m} \approx \sum_{i=1}^N w^s_i \, d\log \theta_i + \frac{1}{2}\sum_{i=1}^N w^s_i \, d^2\log \theta_i,
\end{equation}
since the last term on the right-hand side of \eqref{pfAppendixEq8} cancels \eqref{pfAppendixEq9}. Of course, \eqref{pfAppendixEq10} is equivalent to \eqref{pfAppendixEq100}, which proves \eqref{pfSketchEq3}.
\end{proofThm1}

\vskip 25pt

\section{Supplemental Material} \label{supp}

The decomposition of returns in terms logarithms in Theorem \ref{relValueThm} is convenient and tractable, but not necessary. In particular, it is possible to decompose returns into the same price dispersion and drift components in a purely additive way that does not rely on logarithms.

\begin{thm} \label{relValueThmApp}
Let $F$ be a measure of price dispersion. Then, the portfolio strategy $s = (s_1, \ldots, s_N)$ with
\begin{equation} \label{strategyEq1}
 s_i(t) = \sum_{j=1}^N \theta_j(t)F_j(\theta(t)) - F_i(\theta(t)) + V_s(t)/V_m(t),
\end{equation}
for each $i = 1, \ldots, N$, has a value process $V_s$ that satisfies
\begin{equation} \label{relValueEq1}
 \frac{V_s(T)}{V_m(T)} = \intT\a_F(\theta(t)) - F(\theta(T)),
\end{equation}
for all $T$.
\end{thm}

\begin{proof}
As with Theorem \ref{relValueThm}, Theorem \ref{relValueThmApp} follows from the more general results in Proposition 4.4 of \citet{Karatzas/Ruf:2017}. To see this, let $\tilde{F} = -F$, and note that by definition
\begin{equation*}
 \a_{\tilde{F}}(\theta(t)) = -\a_F(\theta(t)),
\end{equation*}
for all $t$. The function $\tilde{F}$ is regular according to Definition 3.1 of \citet{Karatzas/Ruf:2017} because it is continuous and concave and we have assumed that prices are always positive. Furthermore, because $\tilde{F}$ is twice continuously differentiable, it follows that the finite variation process $\Gamma^{\tilde{F}}$, defined in (3.2) of \citet{Karatzas/Ruf:2017}, satisfies
\begin{equation*}
 d\Gamma^{\tilde{F}}(t) = -\a_{\tilde{F}}(\theta(t)) = \a_F(\theta(t)),
\end{equation*}
for all $t$. Proposition 4.4 then yields the result \eqref{relValueEq1} of Theorem \ref{relValueThmApp}.
\end{proof}

In the same way that Corollaries \ref{returnsCor1} and \ref{returnsCor2} followed from Theorem \ref{relValueThm}, this next corollary follows directly from Theorem \ref{relValueThmApp}.

\begin{cor} \label{returnsCorApp}
The portfolio strategy $g' = (g'_1, \ldots, g'_N)$ with
\begin{equation}
 g'_i(t) = G(\theta(t)) \left( \frac{1}{N\theta_i(t)} - 1 \right) + \frac{V_{g'}(t)}{V_m(t)},
\end{equation}
for each $i = 1, \ldots, N$, has a value process $V_{g'}$ that satisfies
\begin{equation} \label{relValueGMEq1}
 \frac{V_{g'}(T)}{V_m(T)} = - \intT\a_G(\theta(t)) + G(\theta(T)),
\end{equation}
for all $T$. The portfolio strategy $u' = (u'_1, \ldots, u'_N)$ with
\begin{equation}
 u'_i(t) = U(\theta(t)) \big( U^{-\g}(\theta(t))\theta^{\g-1}_i(t) - 1 \big) + \frac{V_{u'}(t)}{V_m(t)},
\end{equation}
for each $i = 1, \ldots, N$, has a value process $V_{u'}$ that satisfies
\begin{equation} \label{relValueEntEq1}
 \frac{V_{u'}(T)}{V_m(T)} = - \intT\a_U(t) + U(\theta(T)) ,
\end{equation}
for all $T$.
\end{cor}

\end{spacing}

\pagebreak

\begin{table}[H]
\vspace{0pt}
\begin{center}
\setlength{\extrarowheight}{3pt}
\begin{tabular} {| l ||c|c|c|}

\hline

      Commodity  & Exchange           & Start      &   Average and Standard Deviation      \\
                          &   where Traded   & Date    &   of Log Price Changes                          \\

\hline

  Soybean Meal      &  CBOT        &  1/1969  & 0.034 (0.303) \\
  Soybean Oil          &  CBOT       &  1/1969  & 0.027 (0.289)  \\
  Soybeans             &  CBOT        &  1/1969  & 0.027 (0.261)   \\
  Wheat                   &  CBOT        &  1/1969  & 0.027 (0.292)  \\
  Corn                     &  CBOT        &  1/1970  & 0.025 (0.260)  \\
  Live Hogs             &  CME          &  1/1970  & 0.022 (0.330)  \\
  Live Cattle            &  CME          &  1/1971  & 0.028 (0.201)  \\
  Cotton                  &  NYBOT      &  1/1973  & 0.018 (0.288)  \\
  Orange Juice       &  CEC           &  1/1973  & 0.029 (0.305)  \\
  Platinum              &  NYMEX      &  1/1973  & 0.041 (0.278)  \\
  Silver                   &  COMEX      &  1/1973  & 0.046 (0.320)  \\
  Coffee                  &  CSC           &  1/1974  & 0.013 (0.360)  \\
  Lumber                &  CME           &  1/1974  & 0.035 (0.326)  \\
  Gold                     &  COMEX     &  1/1975  & 0.045 (0.204)  \\
  Oats                     &  CBOT        &  1/1975  & 0.009 (0.345)  \\
  Sugar                   &  CSC          &  1/1975  & -0.032 (0.408)  \\
  Wheat, K.C.         &  KCBT        &  1/1977  & 0.016 (0.251)  \\
  Feeder Cattle      &  CME          &  1/1978  & 0.028 (0.169)  \\
  Heating Oil          &  NYMEX     &  1/1980  & 0.024 (0.328)  \\
  Cocoa                 &  CSC           &  1/1981  & 0.008 (0.301)  \\
  Wheat, Minn.      &  MGE          &  1/1981  & 0.007 (0.233)  \\
  Palladium            &  NYMEX     &  1/1983  & 0.065 (0.326)  \\
  Crude Oil            &  NYMEX     &  1/1984  & 0.026 (0.354)  \\
  RBOB Gasoline  &  NYMEX     &  1/1985 & 0.034 (0.348)  \\
  Rough Rice         &  CBOT       &  1/1987  & 0.035 (0.277)  \\
  Copper                &  COMEX    &  1/1989  & 0.027 (0.256)  \\
  Natural Gas         &  NYMEX    &  1/1991  & 0.014 (0.515)  \\
  Milk                      &  CME         &  9/1997  & 0.016 (0.277)  \\
  Brent Crude Oil   &  ICE           &  8/2008  & -0.042 (0.332)  \\
  Brent Gasoil        &  ICE           &  8/2008  & -0.047 (0.287)  \\

\hline

\end{tabular}
\end{center}
\vspace{-5pt} \caption{List of commodity futures contracts along with the exchange where each commodity is traded, the date each commodity started trading, and the annualized average and standard deviation (in parentheses) of daily log price changes for each commodity.}
\label{commInfoTab}
\end{table}

\begin{table}[H]
\vspace{15pt}
\begin{center}
\setlength{\extrarowheight}{3pt}
\begin{tabular} {|c||c|c|c|}

\hline

     & Price-Weighted (market)  & Equal-Weighted & CES-Weighted  \\
     & Portfolio                           & Portfolio             & Portfolio           \\

\hline

  1974-2018       &                          3.58\%    \hspace{8pt} (15.15)    &                           6.35\%    \hspace{8pt} (13.60)   &                            7.93\%    \hspace{8pt} (13.75)   \\
  1974-1980       &                          10.94\%  \hspace{3pt} (20.81)    &                           11.97\%  \hspace{3pt} (19.54)   &                            12.55\%  \hspace{3pt} (19.52)  \\
  1980-1990       &                          -2.68\%  \hspace{4pt} (15.00)     &                           1.92\%    \hspace{8pt} (12.62)   &                           4.62\%     \hspace{8pt} (13.40)  \\
  1990-2000       &  \hspace{-10pt}  0.43\%    \hspace{8pt} (7.76)     &  \hspace{-10pt}  2.61\%    \hspace{8pt} (7.24)     &  \hspace{-10pt}  3.62\%     \hspace{8pt} (7.34)  \\
  2000-2010       &                           7.79\%    \hspace{8pt} (16.21)   &                           11.99\%  \hspace{3pt} (14.25)   &                            14.18\%   \hspace{3pt} (13.98)  \\
  2010-2018       &                           1.76\%    \hspace{8pt} (12.97)   &                           3.81\%    \hspace{8pt} (12.65)   &                            5.01\%     \hspace{8pt} (13.05)  \\

\hline

\end{tabular}
\end{center}
\vspace{-5pt} \caption{Annualized average and standard deviation (in parentheses) of monthly returns for price-weighted (market) portfolio and equal- and CES-weighted portfolios, 1974-2018.}
\label{returnsTab}
\end{table}

%
%
%
%
%
%
%
%
%
%

\begin{table}[H]
\vspace{50pt}
\begin{center}
\setlength{\extrarowheight}{3pt}
\begin{tabular} {|c||cc|cc|}

\hline

     &  \multicolumn{2}{c|}{Equal-Weighted} & \multicolumn{2}{c|}{CES-Weighted}  \\
     &  \multicolumn{2}{c|}{Portfolio}             & \multicolumn{2}{c|}{Portfolio}           \\
     & Average (st. dev.) & Sharpe ratio       & Average (st. dev.) & Sharpe ratio    \\

\hline

  1974-2018       &  2.77\%    \hspace{4pt} (4.99)  & 0.55    &    4.34\% \hspace{4pt} (7.01)  & 0.62  \\
  1974-1980       &  1.03\%    \hspace{4pt} (4.03)  & 0.26    &    1.61\% \hspace{4pt} (5.96)  & 0.27 \\
  1980-1990       &  4.60\%    \hspace{4pt} (5.81)  & 0.79    &    7.30\% \hspace{4pt} (8.60)  & 0.85 \\
  1990-2000       &  2.18\%    \hspace{4pt} (2.77)  & 0.79    &    3.19\% \hspace{4pt} (4.34)  & 0.73 \\
  2000-2010       &  4.20\%    \hspace{4pt} (6.11)  & 0.69    &    6.39\% \hspace{4pt} (8.10)  & 0.79 \\
  2010-2018       &  2.06\%    \hspace{4pt} (4.22)  & 0.49    &    3.25\% \hspace{4pt} (5.85)  & 0.56 \\

\hline

\end{tabular}
\end{center}
\vspace{-5pt} \caption{Annualized average, standard deviation (in parentheses), and Sharpe ratio of monthly returns for equal- and CES-weighted portfolios relative to the price-weighted (market) portfolio, 1974-2018.}
\label{relReturnsTab}
\end{table}

\newpage

\begin{figure}[H]
\begin{center}
\vspace{-15pt}
\hspace{-20pt}\scalebox{0.65}{ {\includegraphics{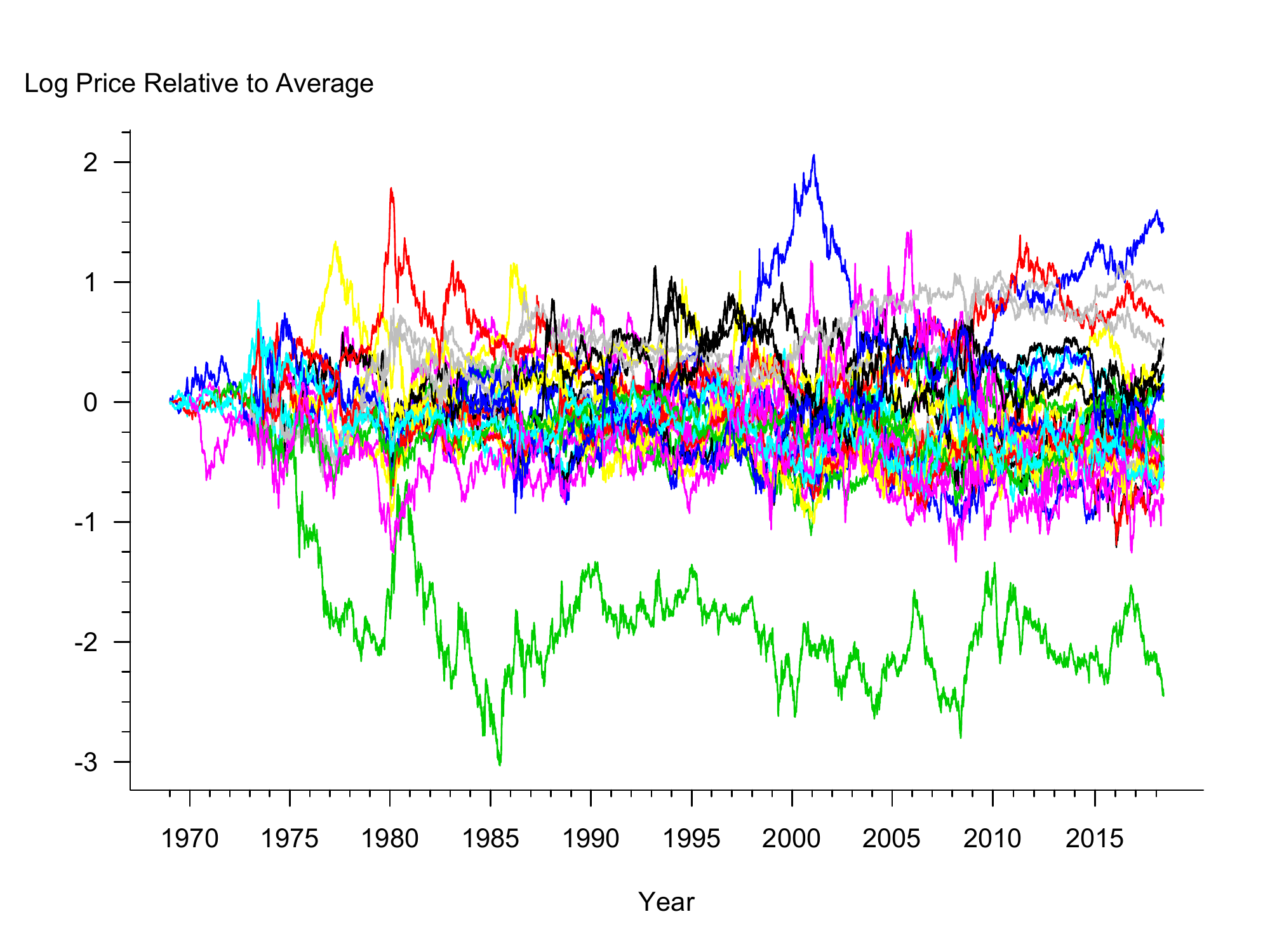}}}
\end{center}
\vspace{-24pt} \caption{Commodity prices relative to the average, 1969-2018.}
\label{relPricesFig}
\end{figure}

\begin{figure}[H]
\begin{center}
\vspace{-4pt}
\hspace{-20pt}\scalebox{0.65}{ {\includegraphics{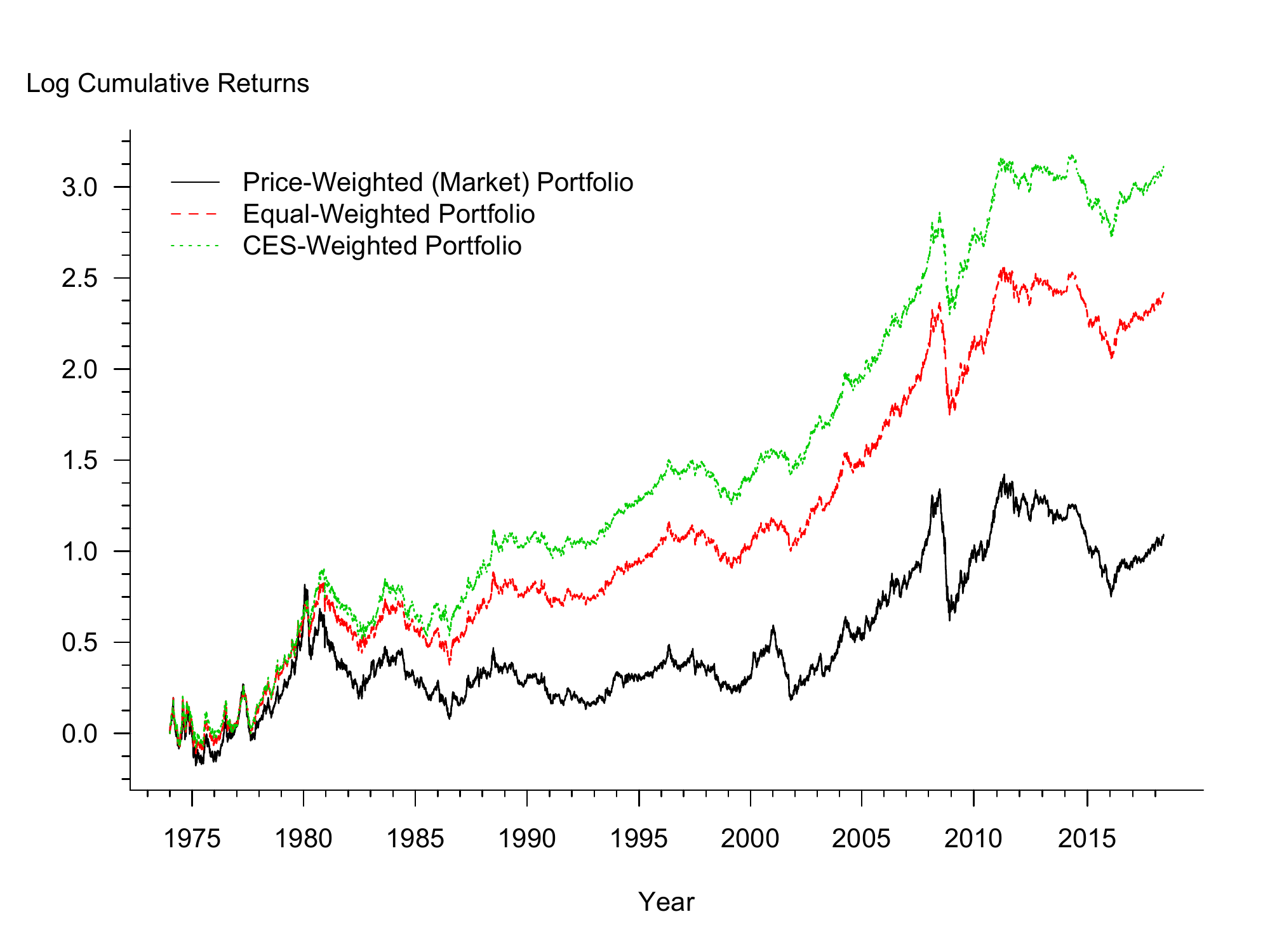}}}
\end{center}
\vspace{-24pt} \caption{Cumulative returns for price-weighted (market) portfolio and equal- and CES-weighted portfolios, 1974-2018.}
\label{returnsFig}
\end{figure}

\begin{figure}[H]
\begin{center}
\vspace{-15pt}
\hspace{-20pt}\scalebox{0.63}{ {\includegraphics{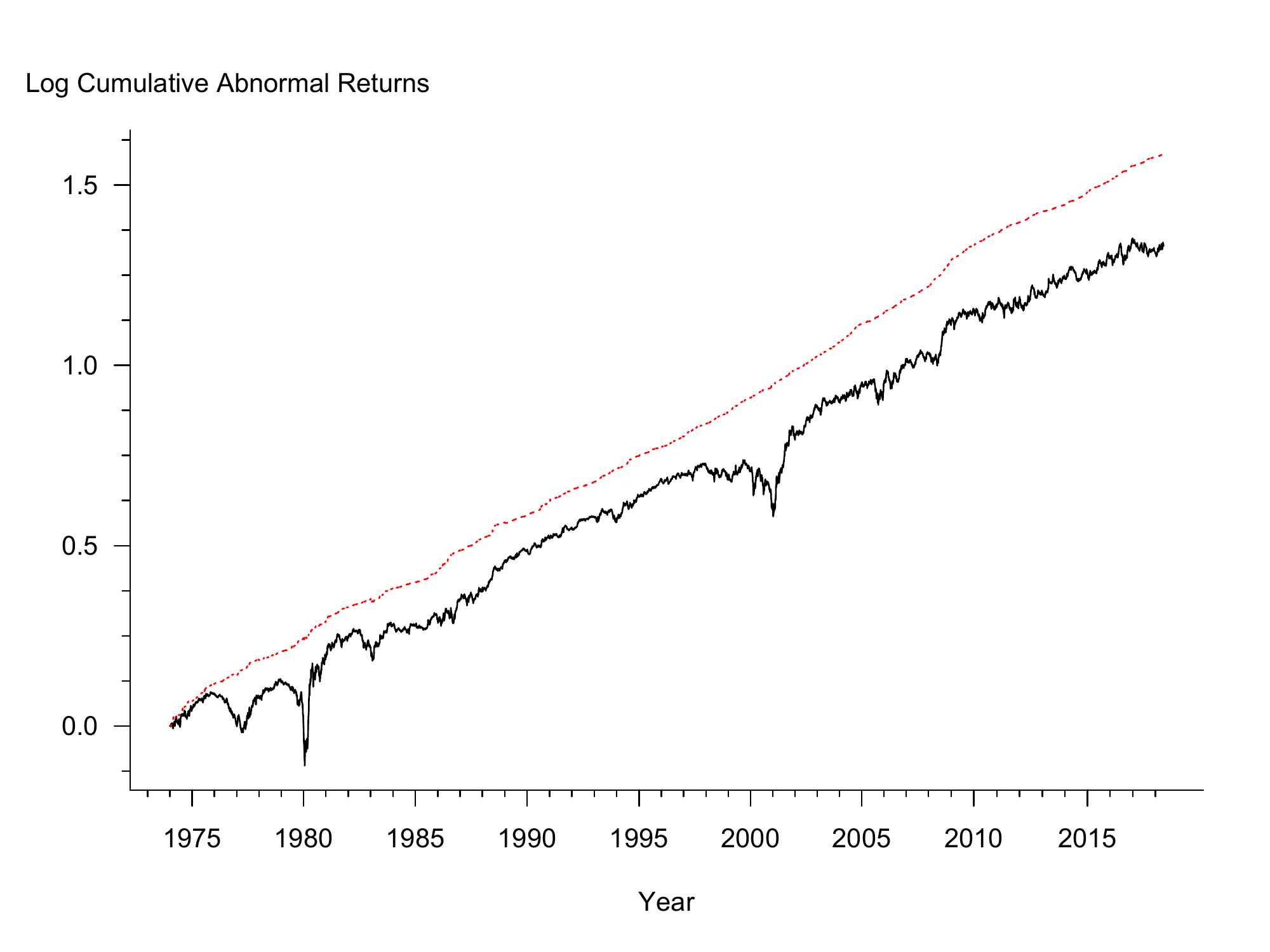}}}
\end{center}
\vspace{-24pt} \caption{Cumulative abnormal returns (solid black line) and adjusted drift (dashed red line) for equal-weighted portfolio, 1974-2018.}
\label{returnsEWFig}
\end{figure}

\begin{figure}[H]
\begin{center}
\vspace{-4pt}
\hspace{-20pt}\scalebox{0.63}{ {\includegraphics{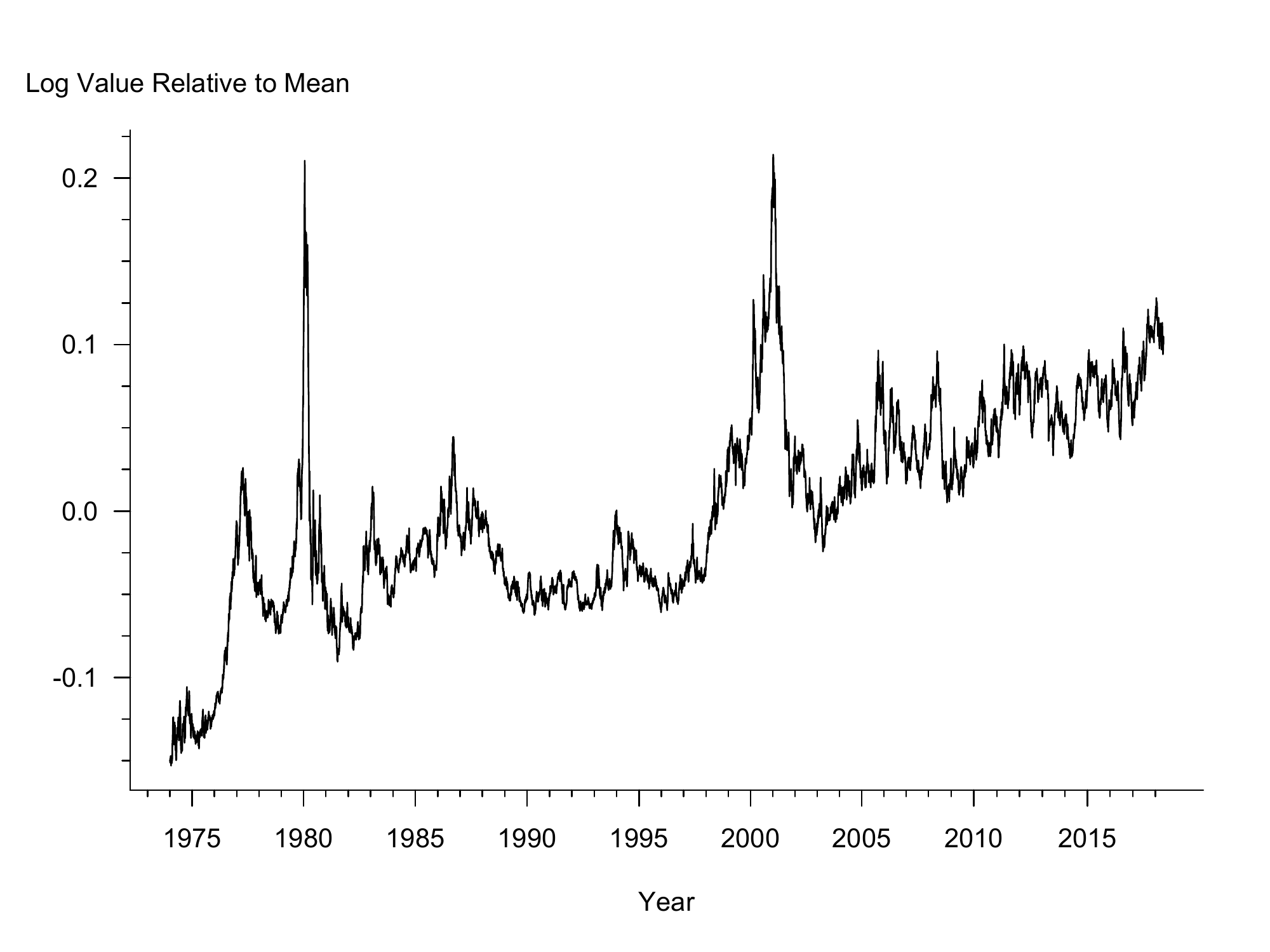}}}
\end{center}
\vspace{-24pt} \caption{Price dispersion for equal-weighted portfolio, measured by minus the geometric mean, 1974-2018.}
\label{dispEWFig}
\end{figure}

\begin{figure}[H]
\begin{center}
\vspace{-15pt}
\hspace{-20pt}\scalebox{0.63}{ {\includegraphics{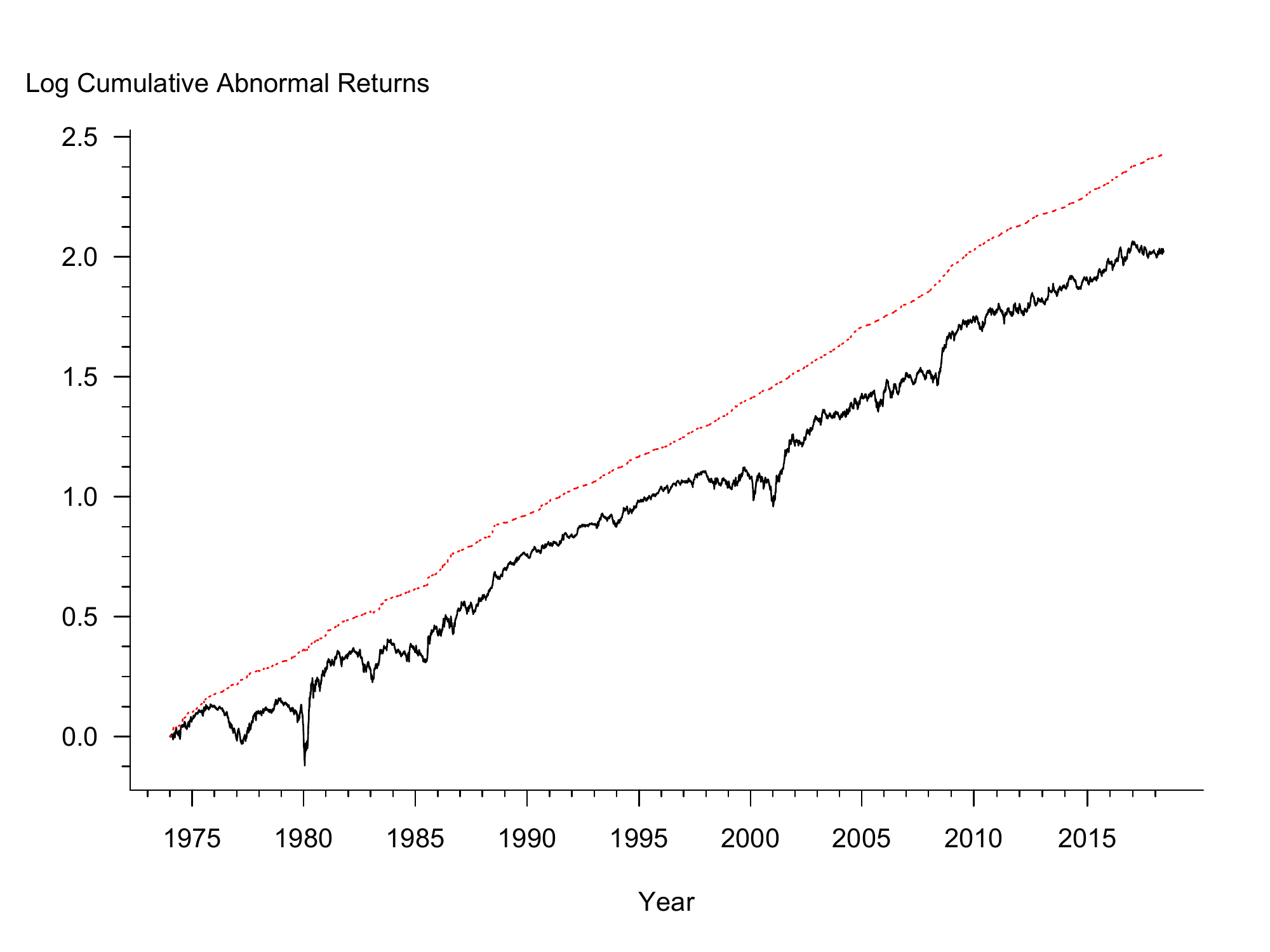}}}
\end{center}
\vspace{-24pt} \caption{Cumulative abnormal returns (solid black line) and adjusted drift (dashed red line) for CES-weighted portfolio, 1974-2018.}
\label{returnsCESFig}
\end{figure}

\begin{figure}[H]
\begin{center}
\vspace{-4pt}
\hspace{-20pt}\scalebox{0.63}{ {\includegraphics{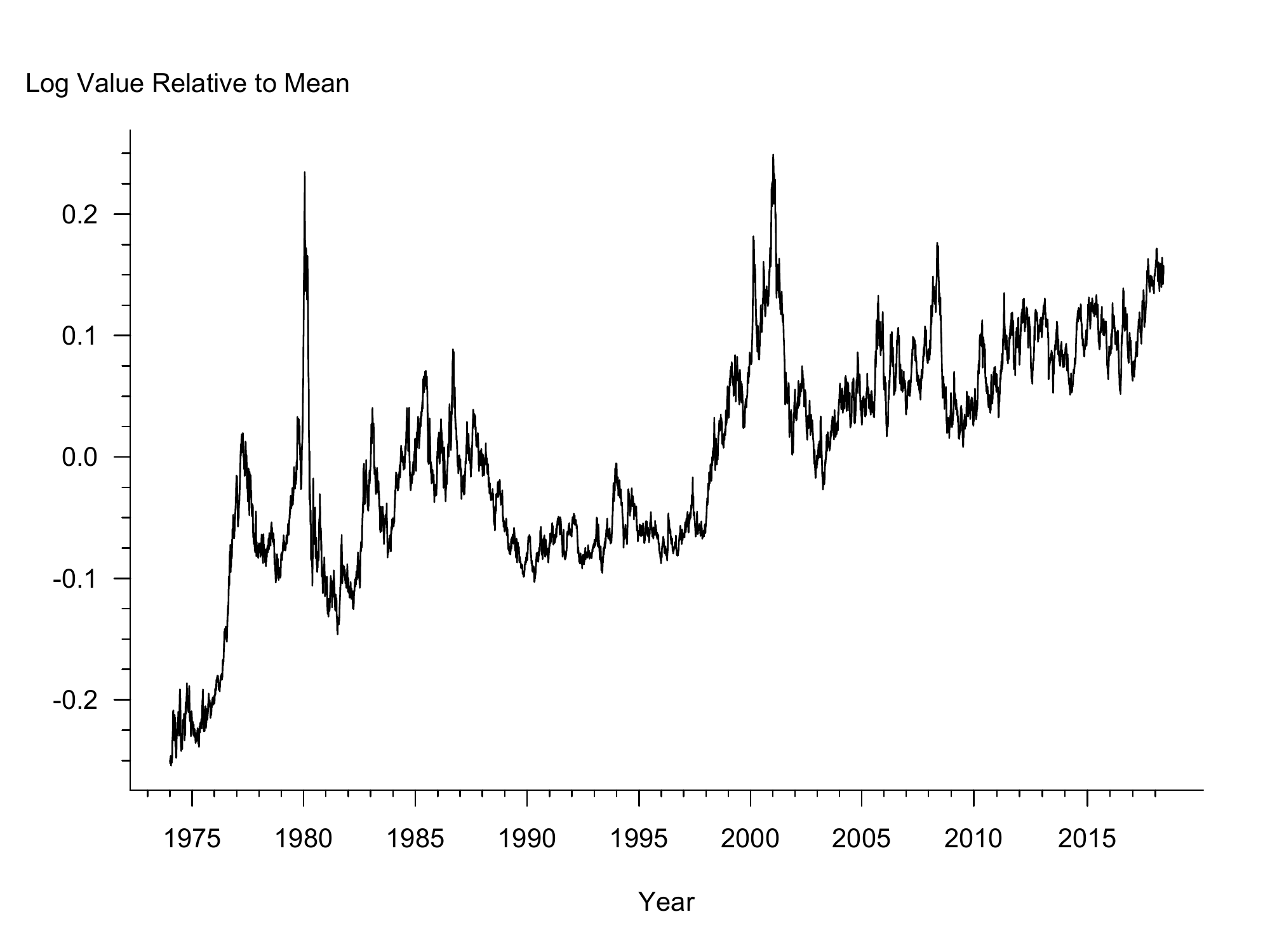}}}
\end{center}
\vspace{-24pt} \caption{Price dispersion for CES-weighted portfolio, measured by minus the CES function, 1974-2018.}
\label{dispCESFig}
\end{figure}

\end{document}